\documentclass[twocolumn]{autart}    

\usepackage{graphicx}      

\usepackage{amsfonts}
\usepackage{amsmath,amssymb}
\usepackage{cite}
\usepackage{color}
\usepackage{enumitem}

\newtheorem{theorem}{Theorem}
\newtheorem{lemma}{Lemma}

\newtheorem{corollary}{Corollary}
\newtheorem{example}{Example}

\newenvironment{proof}[1][Proof]{%
  \par\noindent\textit{#1.} \ignorespaces%
}{%
  \hspace*{\fill}$\square$\medskip%
}

\usepackage{xcolor}

 \usepackage{mathptmx}

 \allowdisplaybreaks
\begin{document}

\begin{frontmatter}

\title{
Contractor-Expander and Universal Inverse Optimal\\ Positive Nonlinear Control }
  
\thanks[footnoteinfo]{This paper was not presented at any IFAC
meeting.}

\author[UCSD]{Miroslav Krstic}\ead{mkrstic@ucsd.edu}

\address[UCSD]{Department of Mechanical and Aerospace Engineering, University of California at San Diego, La Jolla, CA, 92093-0411 USA}

\begin{abstract}
For general control-affine nonlinear systems in the positive orthant, and with positive controls, we show how strict CLFs can be utilized for inverse optimal stabilization. Conventional ``LgV'' inverse optimal feedback laws, for systems with unconstrained states and controls, assume sign-unconstrained inputs and input penalties that are class-K in the input magnitude, hence symmetric about zero. Such techniques do not extend to positive-state-and-control systems. Major customizations are needed, and introduced in this paper, for positive systems where highly asymmetric (or unconventionally symmetric) costs not only on the state but also on control are necessary. With the predator-prey positive-state positive-input benchmark system as inspiration, using a strict CLF built in our previous paper, we prototype two general inverse optimal  methodological frameworks that employ particular ``contractor and expander functions.'' One framework (A) employs a triple consisting of a CLF, a stabilizing feedback, and an expander, whereas the other framework (B) employs a pair of a CLF and a contractor function. Both frameworks yield inverse optimal stabilizer constructions, on positive orthants of arbitrary dimensions. A stronger construction results from a stronger CLF condition.
Biological interpretation for the predator-prey model illuminates that such inverse optimal control constructions  are bio-ecologically meaningful. In addition to general frameworks, we present two fully explicit designs: two Sontag-like universal formulae for stabilization of positive-orthant systems by positive feedback, one of them with inverse optimality. 
\end{abstract}

\end{frontmatter}

\allowdisplaybreaks





 \date{}

\section{Introduction}

In this paper, for which we find little technical precedent in the literature, we solve the {\em inverse optimal} stabilization problems for control-affine nonlinear systems with positive states and controls, given by
\begin{equation}
\label{eq-xi-sys0}
\dot\xi = f(\xi) + g(\xi)\,\omega, \qquad \xi \in (0,\infty)^n,
\end{equation}
with $\omega \in (0,\infty)$ and equilibrium 
\begin{equation}
\xi_{\rm e} = \mathbf{1}:=[1, 1,\ldots,1]^{\rm T}\,, \quad\omega_{\rm e} = 1
\end{equation}
\begin{equation}
f(\mathbf{1}) + g(\mathbf{1}) = 0\,.
\end{equation}
The equilibrium pair $\xi_{\rm e} = \mathbf{1}$,  $\omega_{\rm e} = 1$ in the system \eqref{eq-xi-sys0} is taken without loss of generality (using componentwise state rescaling), in the same sense that $x=0, u=0$ would be taken without loss of generality for conventional Euclidean-state-space systems with real-valued inputs, $\dot x = f(x) + g(x) u, x\in\mathbb{R}^n, u\in \mathbb{R}$. 

Motivated by pedagogy and to  ease the reader into novel concepts, we spend a considerable portion of the first half of the paper on a particular member within the class of system \eqref{eq-xi-sys0}, i.e.,  the predator-prey benchmark system
\begin{subequations}
\label{pp-sf}
\begin{eqnarray}
\label{pp-sf1}
\dot X &=& (1-Y)X
\\
\label{pp-sf2}
\dot Y &=& (X-U)Y
\end{eqnarray}
\end{subequations}
where $X$ is the prey concentration, $Y$ the predator concentration, and $U$ the rate of harvesting of the predator. 

The predator-prey model \eqref{pp-sf}, though stripped of locally destabilizing terms that may induce 
limit-cycling oscillations between the predator and prey concentrations, is a clean and challenging benchmark for control of positive
\eqref{eq-xi-sys0}. In \cite{KrsticAsymmetryDemystified} we designed a globally asymptotically stabilizing positive feedback and a strict CLF for \eqref{pp-sf}. In this paper, with a strict CLF in hand, we progress to the more challenging task of {\em inverse optimal} {\bf\em positive} {\em stabilization}, not tackled before. 

\vspace{-4mm}\paragraph*{$L_gV$ controllers and inverse optimality with sign-unconstrained feedback.} The basic result on inverse-optimal stabilization for control-affine nonlinear systems 
\begin{equation}
\label{eq-xdot=f+gu}
\dot x = f(x) + g(x) u\,, \quad f(0)=0\,,
\end{equation}
with a scalar input goes as follows, by generalizing results from \cite{Krstic1998Stabilization}. Given a CLF $V(x)$, if feedback of the form 
\begin{equation}
u_0 = - \frac{\ell\gamma\left(|L_gV(x)|/\sqrt{r(x)}\right)}{L_gV(x)}\,,    
\end{equation}
where $\ell$ denotes the Legendre-Fenchel transformation (the anti-derivative of the inverse function of the derivative of a function), is stabilizing with respect to $V$ for some strictly positive function $r(x)$ and some $C^1$ class $\mathcal{K}_\infty$ function $\gamma$, whose derivative is also $\mathcal{K}_\infty$, then the ``fortified'' feedback
\begin{eqnarray}
\label{eq-inv-opt-class-u*}
u^* &=& -{\rm sgn}(L_gV) \frac{(\gamma')^{-1}\left(|L_gV(x)|/\sqrt{r(x)}\right)}{\sqrt{r(x)}}
\\
\label{eq-inv-opt-class-u*-altern}
&=& {\rm sgn}(L_gV(x)) \frac{d\left(u_0\  L_gV(x)/\sqrt{r(x)}\right)}{d\left(|L_gV(x)|/\sqrt{r(x)}\right)}
\\
\label{eq-inv-opt-class-u*-altern2}
& =& (\ell\gamma)' \circ \left(\frac{\ell\gamma}{\mathrm{Id}}\right)^{-1}\left(u_0\right)\,,
\end{eqnarray}
where \eqref{eq-inv-opt-class-u*-altern} and \eqref{eq-inv-opt-class-u*-altern2} are forms not previously reported in the literature, is the minimizer of the meaningful cost
\begin{align}
J = \int_0^\infty &\Big[\underbrace{\ell \gamma\left(|L_gV(x)|/\sqrt{r(x)}\right) -L_fV}_{> 0 \ \forall x\neq 0} 
\nonumber\\[2mm]
&+ \gamma\left(\sqrt{r(x)} |u|\right) \Big] dt\,.
\end{align}
If $\gamma(s) = s^2/4$, then 
\begin{equation}
u^*= 2 u_0 = - 2 \frac{1}{r}L_gV\,,
\end{equation}  
has the sign of $-L_gV$ and the cost on control is $r(x)u^2/4$, which is symmetric about $u=0$. 

Such an approach, with sign-indefinite feedback $u^*\gtrless 0$ in \eqref{eq-inv-opt-class-u*}, is clearly inapplicable to positive systems, for which 
costs on both control and the state {\em highly asymmetric} relative to the equilibrium are needed. 

\vspace{-4mm}\paragraph*{Application domains of positive systems and some theory driven by them.} Systems with positive states and inputs arise in many concrete settings where both the quantities being modeled and the actuation mechanisms are intrinsically nonnegative. In ecology and resource management, predator–prey dynamics, fisheries harvesting, and invasive-species control involve population levels and harvesting or stocking rates that cannot be negative. In epidemiology, SIR-type infection models use vaccination, testing, or treatment rates as nonnegative control inputs acting on positive compartments. In chemical and biochemical reactors, species concentrations are positive and manipulated through feed rates or catalyst injections that can only add material. In pharmacokinetics, drug amounts in body compartments are regulated via infusion rates that are strictly nonnegative. In energy and infrastructure systems, reservoir storage levels are controlled by nonnegative inflow rates, and battery charge states are influenced by charging currents constrained to be positive. In traffic and queueing networks, vehicle densities and queue lengths are regulated through metering or service rates that cannot take negative values. In economic growth and capital accumulation models, capital stocks are positive and driven by nonnegative investment rates. These applications share the state and control positivity, making positivity preservation and asymmetric control penalties unavoidable.

Such applications have driven the development of  feedback analysis and designs that respect positivity. In ecological and biochemical models, monotonicity and positivity preservation have motivated the use of monotone control system theory and co-positive Lyapunov functions to guarantee global stability within the positive orthant \cite{AngeliSontag2003,BlanchiniColaneriValcher2020}. In large-scale and networked settings, the  positivity of flows and buffer contents has led to scalable analysis and synthesis tools 
that exploit Perron–Frobenius structure and enable distributed control 
\cite{Rantzer2015}. 
Traffic density control  with nonnegative metering rates has resulted in locally optimal feedbacks that enforce positivity while optimizing throughput \cite{PapageorgiouHajSlSalemBlosseville1991,PapageorgiouReview2018}. In energy and chemical process control, positivity of states and inputs underpins irreversible port-Hamiltonian modeling and thermodynamically consistent stabilization  \cite{RamirezVanDerSchaftMaschke2013}. 

\vspace{-4mm}\paragraph*{Literature relevant to our focus --- positive stabilizing feedback.} Sontag's universal formula for stabilization is the quintessential inverse optimal stabilizer. Two universal formulae exist that are the most relevant to our work---one applicable with positive controls \cite{LinSontag1995} and one for  positive states \cite{Steentjes2018Feedback}, but neither simultaneously. 

For stabilization by {\em positive feedback}, a universal formula was introduced in \cite{LinSontag1995}. However, even though  a logarithmic coordinate change  transforms the positive predator-prey into $\mathbb{R}^2$, covered by \cite{LinSontag1995}, such a formula is not applicable to the CLF here. It is admissible  for stabilization only by overharvesting relative to the equilibrium harvesting rate. Alas, for the predator-prey CLF, there is a triangular set of  predator and prey concentrations (shown later as \eqref{eq-S+})), 
at which both $L_gV>0 , L_fV+L_gV>0$, when only overharvesting is admissible for control. Since positive {\em underharvesting} is also necessary for global stabilization, \cite{LinSontag1995} is inapplicable.

The other reference of considerable relevance to our objective here is \cite{Steentjes2018Feedback}, with a universal formula for stabilization of an equilibrium in the interior of the positive orthant.
But, when applied to our model, the formula from \cite{Steentjes2018Feedback} results in both positive and negative harvesting inputs. 


A pioneering paper on positive linear systems \cite{DeLeenheerAeyels2001} proves, using Brockett's condition, that, for Metzler linear systems with an  equilibrium internal to the positive orthant, asymptotic stabilization is {\em impossible by positive feedback} (by proving that the Perron-Frobenius eigenvalue cannot be moved by positive feedback), and then moves on to stabilization within the positive orthant using sign-unrestricted inputs. Positive stabilization of the nonlinear benchmark predator-prey model \eqref{pp-sf} is not in contradiction with that linear/Metzler result; the directions of the predator-prey drift vector field are not fixed in the positive orthant. 

Other literature  that relates the most closely to our paper is, first and foremost, the author's decade-old \cite{malisoff2016stabilization}, which considers a nearly identical model, but with a strictification that produces a very complex CLF, unappealing for a use in inverse optimal stabilization. Our very recent \cite{11080060} deals with stabilization of a different predator-prey model, where the simultaneous, non-discriminating harvesting of both the predator and the prey is inevitable. We forego here the study of the model with simultaneous harvesting because, while that model presents clear challenges, they are heretofore overcome by ad hoc feedback designs. It is the predator-prey model with single-species harvesting that has a potential for methodological generalizations,  to competition systems and to food chain structures, both via backstepping. 

The inspiration for our work here also comes from the CLF strictification for the SIR dynamics in \cite{9740520}, since our predator-prey dynamics are essentially the (S,I) portion of the SIR, where I\,=\,predator  and S\,=\,prey.

\vspace{-4mm}\paragraph*{Generalizations of the predator-prey result.} The first half of this paper is occupied by designs  for the predator-prey benchmark. These designs not ad hoc. They are  a part of a general framework, extending from Euclidean to positive systems the $L_gV$-based inverse optimal control \cite{Krstic1998Stabilization}. 

While we could have presented the general result first, followed by a construction for the predator-prey model, the predator-prey design is thus far our only fully worked out realization of the paper's general frameworks. We find it, therefore, pedagogically advantageous, to proceed from a fully explicit example to a general framework, i.e., by ``induction,'' rather than ``deductively,'' 
which initially overwhelms, by abstraction. 

We accompany the introduction of our general inverse optimality framework with two concrete general constructions---universal formulae mirroring Sontag's for positive-orthant positive-input systems. Both the constructive and the universal results matter. Universal formulae require constructive designs to supply CLFs. And constructiveness is most clearly realized, conceptually if not practically, through a universal formula.

\vspace{-4mm}\paragraph*{\bf Results, contributions, and ideas introduced:}  
\begin{enumerate}
    \item {\em From strict CLF to inverse optimality.} The paper does not merely present another stabilizer for the predator–prey benchmark, but the first inverse-optimal characterization of a strict CLF for a positive
    system. Relative to our \cite{KrsticAsymmetryDemystified}, where global stabilization of \eqref{pp-sf} is achieved via a strict CLF and a specific backstepping feedback, the present paper establishes that this strict CLF is in fact the value function of an infinite-horizon optimal control problem with a state-dependent and inherently asymmetric penalty on harvesting. 
    \medskip
    \item {\em From optimization penalties symmetric around zero to asymmetric one-sided penalties.} In contrast to classical inverse optimal constructions \cite{Krstic1998Stabilization}, which assume sign-indefinite inputs and symmetric penalties about zero, and unlike the universal positive-input formulas of \cite{LinSontag1995} and the interior-positive equilibrium stabilizers of \cite{Steentjes2018Feedback}, the proposed construction operates under one-sided, unbounded control constraints and a nonzero equilibrium input, and yields a parametrized family of globally stabilizing inverse-optimal feedback laws. 
    \medskip
    \item {\em Expander-based parametrization of asymmetric inverse optimal redesigns.} The novelty lies in the $\Theta\circ\Sigma = \mbox{Id}$  contractor-expander mechanism, which redesigns a stabilizing feedback to make it inverse optimal. A contractor\footnote{We slightly misappropriate established terminology. Our  \emph{contractor} is a function that is not necessarily contractive globally for $s>1$ but is a contraction relative to $s=1$. Any attempt at full precision yields awkward terminology.}
$\Theta$ is a strictly increasing function on a positive real domain with a property that it produces an output smaller than its argument when the argument is larger than unity, and an output larger than one when the argument is smaller than unity.  An expander $\Sigma = \Theta^{-1}$, already lurking in \eqref{eq-inv-opt-class-u*-altern2} but not previously discovered as such, is a function that does the opposite; an inverse function of a contractor. Our functions $\Theta,\Sigma$ simultaneously preserve strict Lyapunov decrease, enforce positivity of control, and induce biologically meaningful asymmetric 
costs. 
    \medskip
    \item {\em Direct constructions of inverse optimal feedbacks with contractor.} A user ready to abandon a nominal feedback  that lacks an $L_gV$ structure and follow a recipe for a direct inverse optimal design, is given such a recipe in the paper. The recipe starts with a a choice of a contractor function in a broad family, and goes through the design of a positive weight-on-control function, which automatically results in a controller with a nonlinear $L_gV$ structure, which ensures inverse optimality. 
    \medskip
    \item {\em Two universal formulae for systems with positive states and controls.} In addition to a general inverse optimality framework, we also provide two universal formulas for positive-input stabilization in the positive orthant. In spite of an uncanny visual similarity with the original universal formula of Sontag \cite{sontag1989universal}, our formulas differ significantly because the domains of its arguments differ, due to the CLF conditions being significantly more restrictive for positive systems relative to unconstrained systems. Our problem with positive-input stabilization on the positive orthant differs significantly from the Lin-Sontag problem on $\mathbb{R}^n$ in \cite{LinSontag1995}, and so do, consequently, the two respective universal formulae. 
\end{enumerate}

\vspace{-4mm}\paragraph*{\bf Organization.} Section \ref{secII} recalls the predator–prey model and the strict CLF from \cite{KrsticAsymmetryDemystified},  including the exact identification of the set $\mathcal{S}_{+}$ on which universal positive-control formulas fail,  
introduces the expanded backstepping stabilizer $U = Y \Sigma(Y/X)$, and establishes global asymptotic stability. Section \ref{secIV} proposes the contractor function $\Theta$ and derives the associated control penalty $\Psi$. Section \ref{secV} proves inverse-optimality via Hamilton–Jacobi–Bellman characterization. Section \ref{secVII} clarifies the connection between the proposed positive inverse optimal feedback and the classical $-r^{-1}L_gV$ structure. Section \ref{secVI} interprets the resulting asymmetric control cost biologically and analytically. Section \ref{secIX} presents simulations illustrating optimality and the role of the contractor mechanism.  Section \ref{secX-} proposes a universal formula for stabilization, a positive-input positive-state analog of Sontag's formula. Sections \ref{secX} and \ref{secX+} give general inverse optimal design for affine-in-positive-control positive-state nonlinear systems. Section \ref{secVIII} explores Sections \ref{secX}'s parametrization of contractor functions and the induced family of control penalties. 

For a reader ready to jump straight into more general but more abstract results, it is entirely feasible to read Section \ref{secX-} first, followed by Sections \ref{secX} and \ref{secX+}, accompanied with the technical Section \ref{secIV}.

\section{Predator-Prey Benchmark:  Model, Backstepping CLF, and Controller Fortification by ``Expander''}
\label{secII}

That the system \eqref{pp-sf} is difficult to stabilize by positive feedback in the positive quadrant is seen by attempting feedback linearization, for example, in the variables $(\ln X, 1-Y)$, and noting that the resulting control $U$ is negative for $Y<1 +\frac{k_1 \ln X - X}{k_2+X}$, where $k_1, k_2>0$. Additionally, even with such a controller, which lapses into negative values, the state still has an excursion out of the positive quadrant. For $k_1=2, k_2=3$, this happens for 
$\bigl(Y_0-1-2\ln X_0\bigr)^2\ge 8\bigl(Y_0-1-\ln X_0\bigr)$ and either $0<X_0\le 1, Y_0> 1+\tfrac{2}{3}\ln X_0$ or $X_0\ge 1, Y_0> 
1+2\ln X_0$. The union of the initial conditions from which either the input or the state start or become negative is given in red in Figure \ref{fig:Red}.

\begin{figure}[t]
    \centering
    \includegraphics[width=0.8\linewidth]{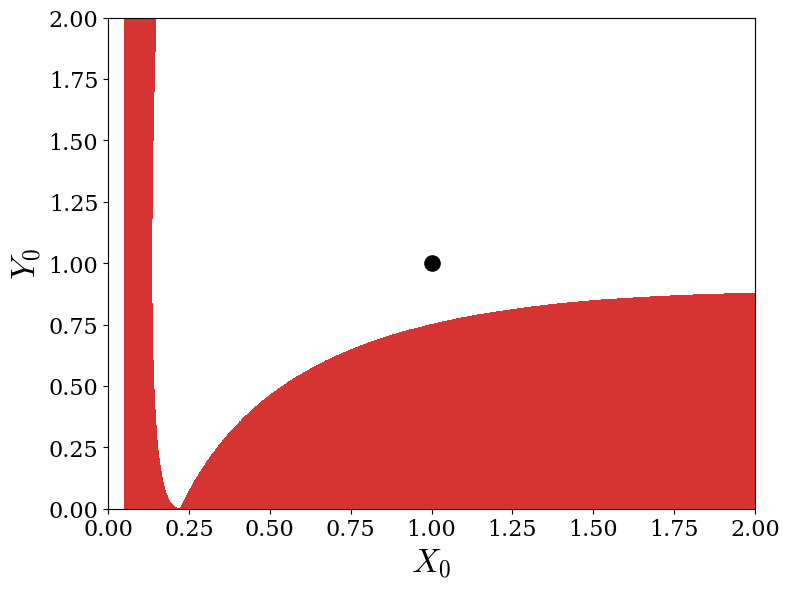}
    \caption{Red: the set of initial conditions $(X_0,Y_0)$ for which feedback-linearizing the predator-prey model results either in the control $U(t)$ or the state $Y(t)$ taking negative values at some $t
    \geq 0$. Black dot: desired equilibrium $(0,0)$.}
    \label{fig:Red}
\end{figure}

The failed feedback linearization attempt exposes the challenge but global asymptotic stabilization of  \eqref{pp-sf} in the positive orthant by positive feedback turns out to be possible. In \cite{KrsticAsymmetryDemystified} we designed the backstepping-based CLF
\begin{equation}
\label{V-bkst}
V(X,Y)
=\Omega\left(\frac{1}{X}\right)+\Omega\left(\frac{Y}{X}\right)\,,
\end{equation}
where
\begin{equation}
\label{eq-Volterra-Lyapunov}
\fbox{$\Omega(S) = S-1 - \ln (S)$}
\end{equation}
so that
\begin{equation}
\label{Vdot-with-L-G}
\dot V
= L(X,Y)+G(X,Y)U
\end{equation}
where
\begin{subequations}
\label{eq-LG-def}
\begin{eqnarray}
\label{eq-L-def}
L(X,Y)&=&\frac{-(X-1)^2+Y(Y-X)}{X},
\\
\label{eq-G-def}
G(X,Y)&=&\frac{X-Y}{X}.
\end{eqnarray}
\end{subequations}
The CLF \eqref{V-bkst}, strict and global on $(0,\infty)^2$ around $X=Y=1$, is unlike anything seen for two-species population systems in the literature before \cite{KrsticAsymmetryDemystified}: not only non-separable in the Volterra-Lyapunov function $\Omega$ but also employing a reciprocal argument.

\begin{figure}[t]
    \centering
    \includegraphics[width=0.9\linewidth]{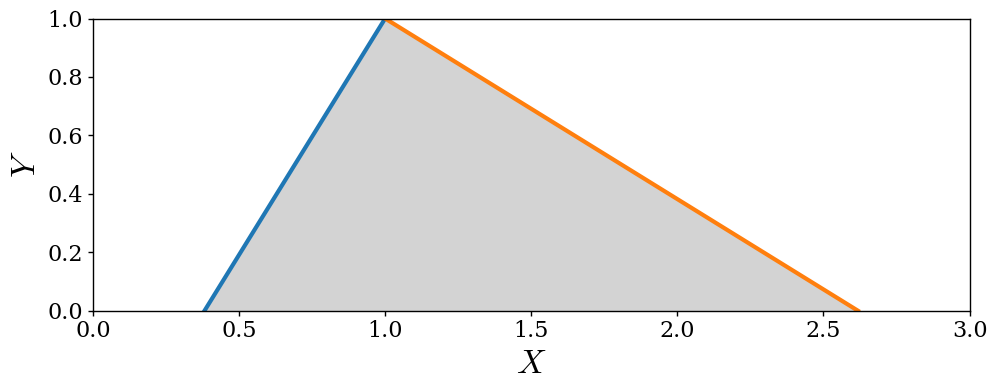}
    \caption{The set $\mathcal{S}_{+}$ in \eqref{eq-S+}, on which the  formula  \cite[(22), (19), (16)]{LinSontag1995} is inapplicable for \eqref{pp-sf}.}
    \label{fig:triangle}
\end{figure}

\begin{rem}\rm
\label{rem-S+}
The universal formula from \cite{LinSontag1995} for systems $\dot x=f(x)+g(x)u$  with $f(0)=0$ and  $u>0$ is inapplicable. With the input shift $u=U-1$ that is equilibrium-anchoring, $u=0$, one gets that $\dot V = (L+G) + Gu$. But $L+G$ and $G$ are  simultaneously positive on the entire triangular subset
\begin{align}
\label{eq-S+}
&\mathcal{S}_{+}
=
\Big\{
(X,Y)\in(0,\infty)^2 :
0<Y<1,\;
\nonumber\\
&\frac{(3-Y)-\sqrt{5}\,(1-Y)}{2}
<
X
<
\frac{(3-Y)+\sqrt{5}\,(1-Y)}{2}
\Big\}\,,
\end{align}
so the  formula  \cite[(22), (19), (16)]{LinSontag1995} neither ensures  stabilization of $X=Y=1$ nor forward invariance of $(0,\infty)^2$.  
\mbox{}\hfill $\triangle$
\end{rem}


\label{secIII}

In \cite{KrsticAsymmetryDemystified} we proposed the backstepping-based globally asymptotically stabilizing positive feedback law
\begin{equation}
\label{eq-nom-fbk-U0}
U_0= \frac{Y^2}{X}\,.
\end{equation}
Next, we introduce the motivating idea for the rest of the paper. We study a redesign of  feedback \eqref{eq-nom-fbk-U0} in which the predator-to-prey ratio is ``expanded'': amplified when $Y>X$ and attenuated when $Y<X$. 
The following theorem is proven in Appendix \ref{pf-Thm1}.

\begin{theorem}[Expander-fortified fbk. remains stabilizing]
\label{thm-intensified}  
Consider the predator-prey system \eqref{pp-sf} under the positive feedback law
\begin{equation}
U(X,Y)=Y\Sigma\!\left(\frac{Y}{X}\right)
\end{equation}
where $\Sigma:(0,\infty)\to(0,\infty)$ is strictly increasing, satisfies $\Sigma(1)=1$, and it expands away from $1$ in the sense that 
\begin{subequations}
\label{eq-expansion}    
\begin{eqnarray}
\label{eq-expansion1}    
\Sigma(s)<s\,, && 0<s<1\\
\label{eq-expansion2}    
\Sigma(s) >s\,, && s>1\,.
\end{eqnarray}
\end{subequations}
Then the equilibrium $X=Y=1$ is globally asymptotically stable on $(0,\infty)^2$.
\end{theorem}

One can come up with numerous examples of function $\Sigma$ that satisfy conditions \eqref{eq-expansion}. The simplest is, arguably, $\Sigma(s)=s^2$, giving $U=Y^3/X^2$. 

To be precise, the function $\Sigma$ performs an expansion for $s>1$ but it acts, mathematically, as a contraction for $s\in(0,1)$. This may confuse, but it is consistent. The reduction of control for $Y/X\in(0,1)$, while it looks like a mathematical contraction, namely, the weakening of control, is, biologically, actually strengthening of control. When the predator is dominated by the prey, reducing the harvesting of the predator actually represents the strengthening of control. Hence, the term ``expander'' is meaningful for $\Sigma$ for both $s\in(0,1)$ and $s>1$. 


\section{Building Control Penalty from Feedback ``Contractor Function''}
\label{secIV}

With the next technical lemma, proven in Appendix \ref{app-proof-Lemma1}, we lay the groundwork for building inverse-optimal cost functions $\Psi(\cdot)$ on control for different feedback laws, using a family of ``contractor functions'' $\Theta(\cdot)$.

\begin{lemma}
\label{lem-Theta-Psi}
Assume $\Theta:(0,\infty)\to(0,\infty)$ satisfies:
\begin{enumerate}
\item $\Theta \in C^1(0,\infty)$ and $\Theta'(s)>0$ for all $s>0$;
\item $\Theta(1)=1$ and 
$\Theta'(1)<1$;
\item $(\Theta(s)-s)(s-1)<0$ for all $s\neq 1$;
\item $\lim_{s\to 0^+}\Theta(s)=0$ and $\lim_{s\to\infty}\Theta(s)=\infty$.
\end{enumerate}
Define $\Psi$ separately on the two sides of $1$ as follows. 
Fix $s_0^- \in (0,1)$ and $s_0^+ \in (1,\infty)$ 
and define
\begin{subequations}
\label{eq-Psidef}
\begin{eqnarray}
\label{eq-Psidef1}
\Psi(s)
&=&
\exp\!\left(
\int_{s_0^-}^{s}\frac{1}{\tau-\Theta(\tau)}\,d\tau
\right), \quad s \in (0,1) \qquad
\\
\label{eq-Psidef2}
\Psi(s)
&=&
\exp\!\left(
\int_{s_0^+}^{s}\frac{1}{\tau-\Theta(\tau)}\,d\tau
\right), \quad s \in (1,\infty)
\\
\label{eq-Psidef3}
\Psi(1)&=& 0\,, \quad \mbox{by continuous extension.}
\end{eqnarray}
\end{subequations}


\noindent Then the following holds:

\medskip

\noindent\textbf{A1.} $\Psi(s)>0$ for all $s\neq 1$, and $\Psi(1)=0$.

\medskip

\noindent\textbf{A2.} $\Psi$ is $C^1$ on $(0,\infty)\setminus\{1\}$ and satisfies the ODE
\begin{equation}
\label{eq-Theta-to-Psi}
{\Psi'(s)}=\frac{1}{s-\Theta(s)}{\Psi(s)}\,, \quad \Psi(1)=0 \,.    
\end{equation}


\noindent\textbf{A3.} $\Psi$ has a strict minimum at $s=1$ and $\Psi(s)\to 0$ as $s\to 1$.

\medskip

\noindent\textbf{A4.} $\Psi$ is strictly decreasing on $(0,1)$ and strictly increasing on $(1,\infty)$.

\medskip

\noindent\textbf{A5.} $\Psi$ is strictly convex on $(0,\infty)\setminus\{1\}$.

\medskip



\end{lemma}

\section{Inverse Optimal Stabilizer for Predator-Prey}
\label{secV}

Now we advance from using an expander $\Sigma$ to strengthen stability to making the feedback inverse optimal. 

\begin{theorem}[Predator-prey inv. opt. stabilizer]
\label{sec-inv-opt}
Let $\Theta$ satisfy the conditions in Lemma \ref{lem-Theta-Psi} and $\Psi$ be defined by \eqref{eq-Psidef}. 
The positive feedback law
\begin{equation}
\label{eq-optimal-U*}
\fbox{$\displaystyle U^*(X,Y)=
Y\, \Sigma\left(\frac{Y}{X}\right)$} \quad X,Y>0 
\end{equation}
where the {\em expander function}
\begin{equation}
\label{eq-def-Sigma}
\Sigma(s):=\Theta^{\rm -1}(s)\,,\quad s>0
\end{equation}
is the 
minimizer of the
infinite-horizon cost
\begin{equation}
J(U)
=
\int_{0}^{\infty}
\left[
q(X,Y)
+
r\,  (X,Y)\,\Psi\left(\frac{U}{Y}\right)
\right]
\,dt\,
\end{equation}
along the solutions of \eqref{pp-sf}, where the state cost
\begin{equation}
\label{eq-q-bkst-def}
q(X,Y)
=
\frac{(X-1)^2}{X}
+
\frac{(Y-X)^2}{X}\frac{Y}{X}\,,
\end{equation}
is positive definite on $X,Y>0$, 
and the state-dependent input weight defined as
\begin{equation}
\label{eq-input-cost-r}
\fbox{$\displaystyle r (X,Y)=-\frac{Y\, G(X,Y)}{\Psi'\left(
\Sigma(Y/X)
\right)}$}
\end{equation}
for $Y\neq X$ and as $r(X,X)=
\lim_{Y\rightarrow X} r(X,Y)$
is positive on $X,Y>0$. 
Furthermore, the minimal cost is given by
\begin{equation}
J^*= V(X_0,Y_0)\,,
\end{equation}
where the value function is given by \eqref{V-bkst}, 
and the optimizer \eqref{eq-optimal-U*} is also a global asymptotic stabilizer of $X=Y=1$ with a region of attraction $X,Y>0$.
\end{theorem}

\proof
We first verify that $r\,  (X,Y)$ is positive on 
$X,Y>0$. 
Since $G(X,Y)=\frac{X-Y}{X}$ and $X>0$, one has $\operatorname{sign}(G)=\operatorname{sign}(X-Y)$. 
Since $Y>0$ and $G(X,Y)=(X-Y)/X$, one has $\operatorname{sign}(-YG)=\operatorname{sign}(Y/X-1)$. Moreover, 
because $\Sigma$ is strictly increasing with $\Sigma(1)=1$, it follows that $\operatorname{sign}(\Sigma(Y/X)-1)=\operatorname{sign}(Y/X-1)$. For a strictly convex $\Psi$ with minimum at $1$, the derivative $\Psi'(s)$ has the same sign as $s-1$, hence $\operatorname{sign}(\Psi'(\Sigma(Y/X)))=\operatorname{sign}(Y/X-1)$. Therefore the numerator and denominator in the expression for $r$ have the same sign, which implies $r(X,Y)>0$ for all $Y\neq X$, $X,Y>0$. 
At $Y=X$,
since $r(X,Y)>0$ for every $Y\neq X$, and since the numerator and denominator in the expression for $r(X,Y)$ have the same sign for all $Y\neq X$, it follows that $r(X,Y)$ remains positive as $Y\to X$. Hence any limit value at $Y=X$ is positive (possibly $+\infty$), and therefore $r(X,X)>0$.

\begin{figure}[t]
    \centering
    \includegraphics[width=0.8\linewidth]{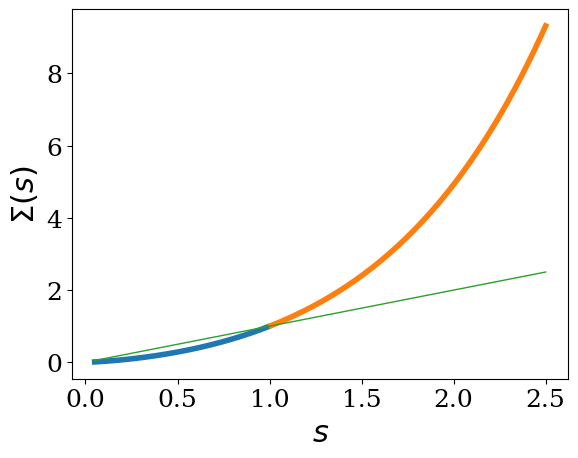}
    \caption{The ``expander'' nonlinearity $\Sigma(s)$, defined by \eqref{eq-def-Sigma}, \eqref{eq-def-Theta}, for $\Psi=\Omega$ defined in \eqref{eq-Volterra-Lyapunov}, and used in the controller \eqref{eq-optimal-U*}. The identity function $s$ in thin green gives the nominal backstepping feedback $U_0=Y^2/X$.}
    \label{fig:Sigma}
\end{figure}

From  differentiation of $V$ along trajectories of \eqref{pp-sf}, we have \eqref{Vdot-with-L-G}, where $L,G$ are defined by \eqref{eq-L-def}, \eqref{eq-G-def}. 
Define the Hamiltonian
\begin{eqnarray}
\label{eq-Hamiltonian}
H(X,Y,U)
&=&
q(X,Y)
+
L(X,Y)
+
G(X,Y)\,U
\nonumber\\
&& +
r\,  (X,Y)\,\Psi(U/Y)\,,
\end{eqnarray}
where
\begin{equation}
\label{eq-q-def-alternative}
q(X,Y)
=-\Big(L(X,Y)+G(X,Y)\frac{Y^2}{X}\Big)\,,
\end{equation}
with $L,G$  defined by \eqref{eq-L-def}, \eqref{eq-G-def}, is equal to \eqref{eq-q-bkst-def}, and  evident to be strictly positive for all $(X,Y)\neq(1,1)$ and vanishes  at $(1,1)$.
For fixed $(X,Y)$, the mapping
$U\mapsto G\,U+r\,  \,\Psi(U/Y)$
is strictly convex on $(0,\infty)$ because $\Psi$ is strictly convex and $r\,  >0$.
The mapping's derivative is
\begin{equation}
\frac{\partial}{\partial U}
\big(GU+r  \Psi(U/Y)\big)
=
G+\frac{r}{Y}   \Psi'(U/Y).
\end{equation}
By definition of $r$ in \eqref{eq-input-cost-r},
$G+\frac{r}{Y}   \Psi'(U^*/Y)=0$, 
so $U^*$ is the unique global minimizer.
Substituting $U=U^*$ into \eqref{eq-Hamiltonian},
\begin{equation}
\label{eq-Ham1}
H(X,Y,U^*)
=q+L+ G U^* +r\,  \Psi(U^*/Y).
\end{equation}
With $r$ in \eqref{eq-input-cost-r},
we obtain
\begin{eqnarray}
G U^*+r\,  \Psi(U^*/Y)
&=&
GY\left(U^*/Y-\frac{\Psi(U^*/Y)}{\Psi'(U^*/Y)}\right)
\nonumber\\
&=&
G Y\,\Theta(U^*/Y)
=
G\,\frac{Y^2}{X}.
\end{eqnarray}
Hence
\begin{equation}
\label{eq-Ham2}
H(X,Y,U^*)
=
q+L+G\frac{Y^2}{X}
=
0
\end{equation}
by definition of $q$ in \eqref{eq-q-def-alternative}. Therefore
$\min_{U>0} H(X,Y,U)=0$, 
so the stationary Hamilton--Jacobi--Bellman equation holds with value function $V$.

To prove global asymptotic stability, we can simply invoke Theorem \ref{thm-intensified} by noting that $\Sigma$ defined in \eqref{eq-def-Sigma} satisfies conditions \eqref{eq-expansion}. But we provide a proof customized to the optimal control result of this theorem. Recall from \eqref{eq-Ham1} and \eqref{eq-Ham2} that 
\begin{equation}
q+L+G\,U^*(X,Y)
+r\,\Psi\!\left(U^*/Y\right)=0,
\end{equation}
which gives
\begin{equation}
\dot V
=
L+G U^*
=
-q
- r\,\Psi\!\left(U^*/Y\right).
\end{equation}
Since, on $X,Y>0$, $q$ is positive definite at $X=Y=1$, $r$ is positive, and $\Psi\!\left(U^*/Y\right)$ nonnegative, the global asymptotic stability conclusion follows. 
\hfill $\square$



\section{Recognizing $L_gV$ Control in the Inverse Optimal Feedback for Predator-Prey Dynamics} 
\label{secVII}

Recognizing the expected $-\dfrac{1}{r} L_gV$ structure, with some positive state-dependent gain $r$, in the feedback law \eqref{eq-optimal-U*}, is  difficult. First, what is $L_gV$ there? From \eqref{Vdot-with-L-G} it appears like $G = 1-Y/X$ should play the $L_gV$ role.
But one has to recognize that, as the actual input to the predator-prey model, one should regard not $U$, but the predator-scaled and shifted-to-setpoint-1 input $u=U/Y - 1$. 

To facilitate the recognition of the $L_gV$ format of the nominal and optimal controllers, we first specialize Theorem \ref{sec-inv-opt} to the Volterra-Lyapunov control cost $\Omega$ given in \eqref{eq-Volterra-Lyapunov}.

\begin{corollary}[{\em Volterra}-cost optimality for predator-prey]
\label{cor-Volterra}
For
\begin{equation}
    \Psi=\Omega\,,
\end{equation}
with $\Omega$ defined in \eqref{eq-Volterra-Lyapunov}, in which case
\begin{equation}
\label{eq-def-Theta-Volterra}
\Theta(s)
=\frac{s\ln s}{s-1}
\end{equation}
and
\begin{eqnarray}
\label{eq-input-cost-r-Volterra}
r  (X,Y) 
&=& \Pi(Y/X)Y
\\
\Pi(Y/X) &:=&\frac{\rho -1}{\Sigma\left(\rho \right)-1} \Sigma\left(\rho \right) \,, \quad \rho = \dfrac{Y}{X}\,,
\end{eqnarray}
with $\Sigma$ given by \eqref{eq-def-Sigma}, the results of Theorem \ref{sec-inv-opt} hold.
\end{corollary}

For the ``input'' $u$, for the control cost \eqref{eq-def-Theta-Volterra} and with the recognition that $(\Psi')^{-1}=1/(1-s)$, and with the expander \eqref{eq-def-Sigma}, the nominal-optimal control pair is expressed by
\begin{eqnarray}
\label{eq-u01}
u_0&=& \frac{U_0}{Y} - 1 =\frac{\ln\left({1+\dfrac{GY}{r}}\right)}{\dfrac{GY}{r}} -1
\\
\label{eq-u*1}
u^*&=& \frac{U^*}{Y} - 1 =-\frac{\dfrac{GY}{r}}{1+\dfrac{GY}{r}} 
\,.
\end{eqnarray}
Their relation is expressed by 
\begin{equation}
\label{eq-u0-to-u*}
u^* = \frac{d \left(u_0\ G Y/\sqrt{r} \right)}{d\left(G Y/\sqrt{r} \right)}\,,
\end{equation}
consistent with the nominal-optimal relation \eqref{eq-inv-opt-class-u*-altern}. 
The actual nominal-optimal control pair is 
\begin{eqnarray}
\label{eq-u02}
U_0 &=& Y \frac{\ln\left({1+\dfrac{GY}{r}}\right)}{\dfrac{GY}{r}}>0
\\
\label{eq-u*2}
U^* &=& Y \frac{1}{1+\dfrac{GY}{r}}>0\,.
\end{eqnarray}

\begin{figure}[t]
    \centering
    \includegraphics[width=0.8\linewidth]{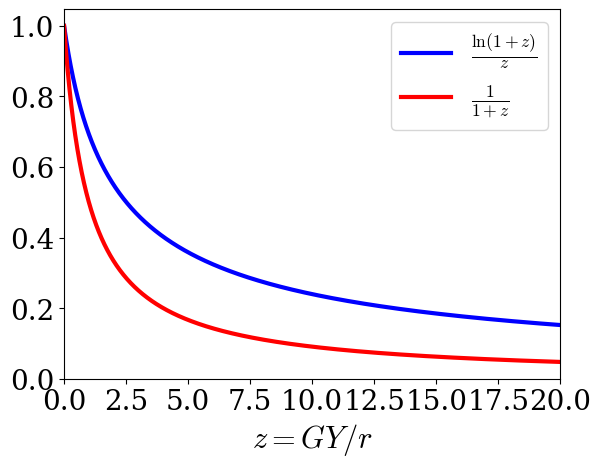}
    \caption{The functions that govern $U/Y$ in the feedback laws \eqref{eq-u02} and \eqref{eq-u*2}. Compared to the blue curve, for the nominal control, the red curve, representing the optimal control, decays {\em twice} as fast near $z=0$, a tell-tale sign of optimality (the gain margin of 1/2), and at a rate that gets even faster --- $\ln(z)$ faster --- as $z$ grows.}
    \label{fig:GYr}
\end{figure}

Figure \ref{fig:GYr} confirms the experience with $L_gV$ controllers but requires an interpretation. The optimal control is amplified with the expender function $\Sigma(\cdot)$ as its argument $Y/X$ grows above 1 (and, likewise, attenuated, when $Y/X$ decays below 1), as shown in Figure \ref{fig:Sigma}. However, in terms of the dependence on  argument $G(X,Y)Y/r(X,Y)$, the two controllers are decreasing, and the optimal control $U^*/Y$ decreases much more aggressively than the nominal control $U_0/Y$ away from its equilibrium value $U_{\rm e}=1$ with the increase of $GY/r$. This can potentially confuse since $U^*$ grows more aggressively than $U_0$ in $Y/X$, and both $U^*, U_0$ can exceed $1$, but is explained by the factor $Y$, which can be both (much) larger and smaller than $1$.

Though with less transparency, for a general convex $\Psi$, the controllers are written as
\begin{eqnarray}
U_0 &=& Y\Theta\left((\Psi')^{-1}\left( -\dfrac{GY}{r}\right) \right)
\\
U^* &=& Y(\Psi')^{-1}\left( -\dfrac{GY}{r}\right)\,,
\end{eqnarray}
and the relation \eqref{eq-u0-to-u*} is obeyed even for the general $\Psi$ (with $u_0=U_0/Y-1, u^*=U^*/Y-1$). 


\section{Biological Interpretation of Inverse Optimality}
\label{secVI}

The particular inverse optimal result of Corollary \ref{cor-Volterra} is rich with biologically interesting  justification. 

\vspace{-3mm}\paragraph*{Asymmetry of optimal feedback: high-predator amplification and high-prey attenuation in predator harvesting.} Let us first consider the optimal feedback \eqref{eq-optimal-U*} and the effect of the strictly increasing low-attenuator high-amplifier function $\Sigma(\cdot)$ defined with \eqref{eq-def-Theta-Volterra} and \eqref{eq-def-Sigma}. Near $s=0^+$, $\Sigma(s)\sim s/|\ln(s)|$, which has a zero slope but its slope is higher than $s^\alpha$ for all $\alpha>1$, i.e., it is locally superlinear but less flat than any locally superlinear polynomial. So, $\Sigma$ is a low-$Y/X$ attenuator. It uses mild harvesting of the predator when the prey dominates the predator. As $s\rightarrow\infty$, $\Sigma(s)\sim {\rm e}^s$. It is a high-$Y/X$ amplifier, namely, it harvests the predator aggressively when the predator dominates the prey. Finally, $\Sigma(1)=1$ and $\Sigma'(1)=2$. The latter property is a tell-tale sign of inverse optimality and the resulting use of twice the gain than is needed for mere stabilization, which results in a gain margin of $[1/2, \infty)$. 

\vspace{-3mm}\paragraph*{Predator-weighted and asymmetric control penalty.} We consider the complete control penalty term, $r(X,Y) \Psi(U/Y)$.
Near $s=0^+$, $\Pi(s)\sim \Sigma(s) \sim s/|\ln(s)|$, meaning $\Pi(\cdot)$ at zero is zero and flat. As $s\rightarrow\infty$, $\Pi(s)\sim s$. And, $\Pi(1)=1/2, \Pi'(1)=2/3$. Let us now consider the sensitivity of the state-dependent control cost to harvesting $U$:
\begin{equation}
\frac{\partial}{\partial U}
\left[
r(X,Y)\,\Psi\!\left(\frac{U}{Y}\right)
\right]
=
\Pi\!\left(\frac{Y}{X}\right)
\left(1-\frac{Y}{U}\right).
\end{equation}
This is a product of two functions, each increasing in its single arguments, $Y/X$ and $U/Y$. The cost sensitivity can be studied. We focus on three cases that are not immediately intuitive, biologically. 

In the regime $X \gg Y \gg U$, prey are relatively abundant compared to predators, the predator is temporarily scarce in relative terms, and harvesting per predator is extremely low. A purely predator-centric heuristic would suggest reducing harvesting even further. However, in a predator–prey system where prey are relatively abundant, the predator has strong intrinsic recovery potential; shutting down harvesting entirely does not enhance survival but instead risks destabilizing the prey–predator balance. The separable cost sensitivity to harvesting, $\partial[r(X,Y)\Psi(U/Y)]/\partial U \sim - \frac{Y/X}{|\ln(Y/X)|}\frac{Y}{U}<0$, reflects this system-level perspective: it does not reward driving harvest intensity $U/Y$ toward zero simply because the predator population is momentarily small relative to prey, since the ecological context (relative prey abundance) already provides recovery support.

\begin{figure}[t]
    \centering
    \includegraphics[width=0.8\linewidth]{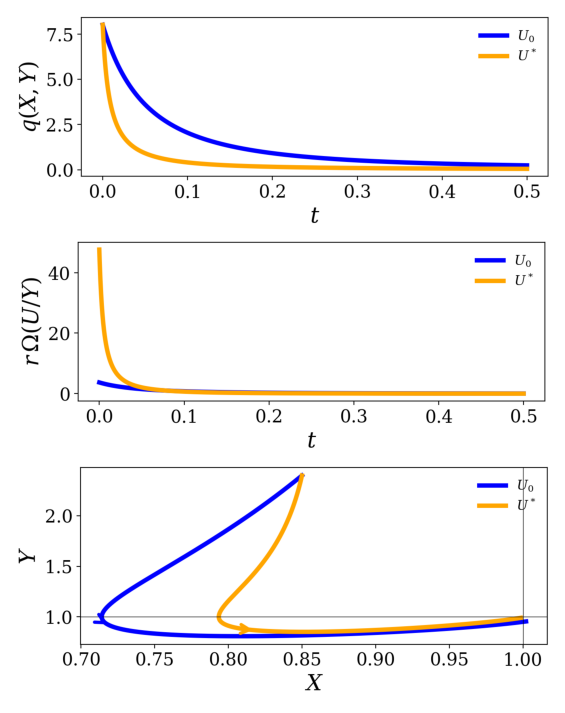}
    \caption{With initial prey slightly depleted and predator dominant, the expander $\Sigma$ makes the action of optimal control $U^*=Y \Sigma(Y/X)$ stronger than the dominant control $U_0=Y^2/X$ and reduces the excursion of the trajectory on the way to the equilibrium $X=Y=U=1$. Perhaps counterintuitively, the optimal controller, which harvests the predator more aggressively, prevents a deeper depletion of the prey than the nominal controller. While the control cost $r(X,Y) \Omega(X,Y)$ is higher for the optimal control initially, the state cost $q(X,Y)$ is lower for the duration of the transient, resulting in the overall cost $J^*$ that is lower than $J$ achieved under the nominal control, as predicted by the theory, and this reduction happens to be considerable: 20\%.}
    \label{fig:pred>prey}
\end{figure}

In the opposite regime, $U\gg Y \gg X$, the predator dominates the prey and is already being harvested at a high per-capita rate. One might expect that, since the predator is abundant, further harvesting should remain inexpensive. Instead, the marginal sensitivity of the cost on control, $\partial[r(X,Y)\Psi(U/Y)]/\partial U \sim Y/X\gg 1$, becomes large and positive because $\Pi(Y/X)$ is large while $1 - Y/U$ is close to its saturation value of one. The model therefore reacts strongly against additional increases in harvesting. In ecological terms, once harvesting intensity per predator is already high, the inverse optimal formulation treats further escalation as dangerous, even under predator dominance, thereby preventing overcorrection that could induce a deep depletion in the prey population.

When $Y\ll X$, $r(X,Y)$ is small, i.e., it appears that the inverse optimal formulation makes it cheap to harvest a relatively scarce predator. That is indeed the case and in fact reasonable: in this regime the abundance of the prey has a greater effect on the predator recovery and the system as a whole than the harvesting does.

\section{Simulations: Illustration of Optimality and Effect of the Expander}
\label{secIX}

The purpose of the simulations in this section is not to demonstrate stability --- which is already guaranteed by Theorems~\ref{thm-intensified} and \ref{sec-inv-opt} --- but to illustrate how inverse optimality manifests itself dynamically, and how the expander function $\Sigma$ reshapes the transient response relative to the nominal backstepping controller $U_0 = Y^2/X$.

Two initial conditions are chosen to provide insight.

\begin{figure}[t]
    \centering
    \includegraphics[width=0.8\linewidth]{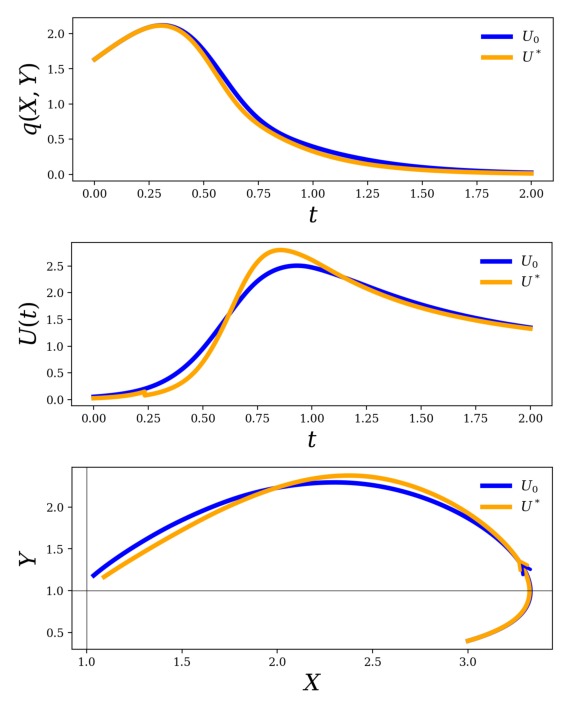}
    \caption{The initial condition with a depleted predator and an overpopulated prey is particularly challenging for the optimal controller because the initial expander $\Sigma(Y_0/X_0)= 0.16$ is half of  $Y_0/X_0=0.35$, making the initial optimal harvesting ``weak.'' The transients do not differ by much, but the control effort $U(t)$ differs, as expected: the optimal control is first weaker and later stronger than the nominal control. As expected from the theory, the overall cost $J^*$ of the optimal controller is lower, but only by about $3\%$, in this case, chosen to challenge the optimal controller.}
    \label{fig:prey>predator}
\end{figure}

First, we consider predator-dominant regimes with the prey initially depleted. In this case $Y/X > 1$, so the expander satisfies $\Sigma(Y/X) > Y/X$ and amplifies harvesting intensity. This regime highlights the aggressive behavior of the optimal controller and tests whether expanded harvesting prevents deeper prey depletion. Figure~\ref{fig:pred>prey} shows that the optimal feedback produces a visibly stronger initial control action, reduces the excursion of the trajectory in the $(X,Y)$-plane, and yields a substantially lower total cost $J^*$, in agreement with the Hamilton--Jacobi--Bellman characterization.

Second, we examine prey-dominant regimes with a depleted predator. Here $Y/X < 1$ and the expander $\Sigma$ attenuates the harvesting intensity, $\Sigma(Y/X) < Y/X$. This configuration is  selected as a stress test for the optimal controller: attenuation weakens the initial harvesting action relative to $U_0$, so any performance advantage must arise from the state-dependent weighting structure rather than brute-force amplification. As seen in Fig.~\ref{fig:prey>predator}, the state transients are nearly indistinguishable, while the control histories differ markedly in early time. The total cost improvement is smaller but remains strictly positive, consistent with Theorem~\ref{sec-inv-opt}.


\section{Two universal Formulae for Stabilization of Positive-Orthant Positive-Control Systems }
\label{secX-}

In this section we make our first step towards generalizing the framework from the predator-prey benchmark to broader systems affine in control. Under two different CLF assumptions, one less and the other more restrictive, we provide universal formulas for stabilization in the positive orthant by positive control. It is only for a more restricted class of systems that one can produce a universal formula that achieves not only stabilization but also inverse optimality.

\begin{theorem}[Univ. formula for stabilization---{\em general}]
\label{thm-univ0}
Consider the system 
\begin{equation}
\label{eq-xi-sys}
\dot\xi = f(\xi) + g(\xi)\,\omega, \qquad \xi \in (0,\infty)^n,
\end{equation}
with a positive scalar input $\omega \in (0,\infty)$ and 
such that
$f({\bf 1})+g({\bf 1})=0$.
Let $V:(0,\infty)^n\to\mathbb{R}_{\ge0}$ be $C^1$, positive definite with respect to ${\bf 1}$, proper in $(0,\infty)^n$, and let it satisfy the CLF condition
\begin{equation}
\label{eq-CLF-true}
L_gV(\xi)\ge 0 \;\Rightarrow\; L_fV(\xi)<0,
\qquad \forall \xi\in (0,\infty)\setminus {\bf 1}.
\end{equation}
The positive-valued feedback defined by $\omega_{\rm u}({\bf 1}) = 1$ and 
\begin{eqnarray}
\label{eq-univ-basic}
\omega_{\rm u}(\xi)
&=& 1- \frac{L_fV+ \sqrt{(L_fV)^2+(L_gV)^2}}{L_gV} 
\nonumber\\
& =& 1 - \frac{L_gV(\xi)}{r(\xi)} >0 
\,, \quad \xi\neq {\bf 1}
\end{eqnarray}
where
\begin{equation}
r = -L_fV+ \sqrt{(L_fV)^2+(L_gV)^2} > 0 \,, \quad \xi\neq {\bf 1}
\end{equation}
guarantees
\begin{eqnarray}
\dot V&=&
L_gV-\sqrt{(L_fV)^2+(L_gV)^2}
\nonumber\\
&=&-\frac{(L_fV)^2}{L_gV +\sqrt{(L_fV)^2+(L_gV)^2}}
<0,
\qquad \forall \xi\neq {\bf 1},
\end{eqnarray}
so the equilibrium $\xi={\bf 1}$ is globally asymptotically stable on $(0,\infty)^n$.
\end{theorem}

\begin{proof}
The expression for $\dot V$ is obtained by direct computation and its negative definiteness is evident, with $\dot V=0$ only when $L_fV=L_gV=0$, i.e., at $\xi={\bf 1}$. The control is positive because
${L_gV}<{r}$,
which follows from 
\begin{equation}
r-L_gV
=
\sqrt{(L_fV)^2+(L_gV)^2}-L_fV-L_gV
\,,\end{equation}
 the fact that, on the admissible set, either 
$-L_gV> 0 $ or $ -L_fV>0$, and tha fact that $\sqrt{(L_fV)^2+(L_gV)^2}-L_fV>0$ and $\sqrt{(L_fV)^2+(L_gV)^2}-L_gV>0$.
\end{proof}

Universal formula \eqref{eq-univ-basic} shifted as $\omega-1$ applies under the same CLF condition \eqref{eq-CLF-true},  ensures global asymptotic stabilization of ${\bf 1}$ in $(0,\infty)^n$, and is immensely simpler than \cite[(22), (19), (16]{LinSontag1995}. 

\begin{rem}
[Universal formula for predator-prey.]
\rm
\label{rem-pred-univ}
With \eqref{eq-LG-def} substituted into \eqref{eq-univ-basic}, we obtain
\begin{equation}
\label{eq-U-univ}
U_{\mathrm{u}}(X,Y)
=
1-\frac{G}{\sqrt{L^2+G^2}-L}\,,
\end{equation}
shown in Figure \ref{fig:U_univ}.
Along with the nominal $U_0= Y^2/X$ and, say, the simple inverse optimal $U^*=Y^3/X^2$ (obtained for $\Sigma(s)=s^2$)
? Unclear, \eqref{eq-U-univ} adds to the catalog of options. 
\hfill $\triangle$
\end{rem}

\begin{figure}[t]
    \centering
    \includegraphics[width=0.8\linewidth]{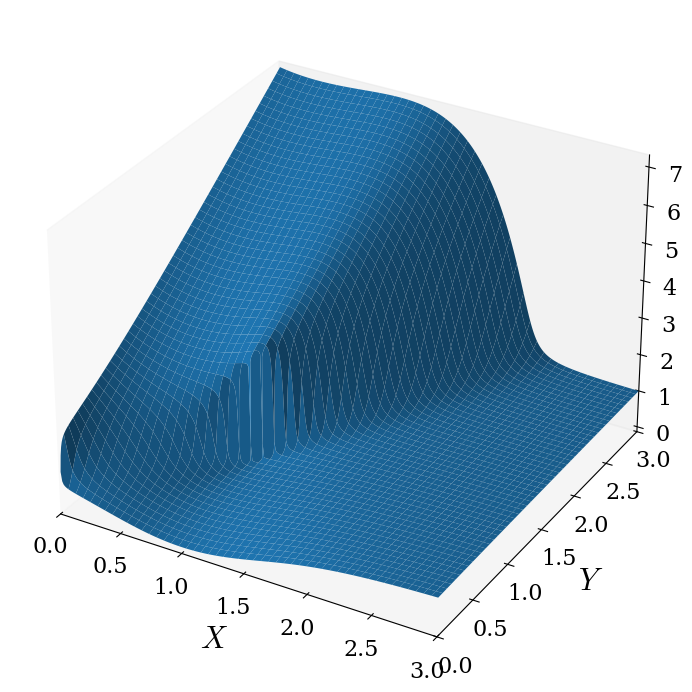}
    \caption{Universal control for predator-prey, $U_{\rm u}(X,Y)$  in \eqref{eq-univ-basic}. Harvesting $U_{\rm u}$ reaches its minimum = 0 at the extinct-predator state $Y=0$ (outside of the state space) and  equilibrium prey state $X=1$. Control $U_{\rm u}(X,Y)$ appears piecewise defined but this is the result of the denominator in \eqref{eq-U-univ} and the delicate cancellation at $X=Y=0$. Control $U_{\rm u}<1$ on the set $\{X<Y\} \supset \mathcal{S}_{+}$ shown in Fig. \ref{fig:triangle}.}
    \label{fig:U_univ}
\end{figure}

The stabilizing feedback  \eqref{eq-univ-basic}, even though it has $L_gV$ as a factor, is not inverse optimal with respect to a positive definite cost on $\xi$. This is disappointing, since Sontag's formula, and many of its extensions, are in fact inverse optimal. 

To show that it is possible to achieve an inverse optimal universal formula, we sacrifice the generality of the class of systems and trade the non-restrictive necessary-and-sufficient CLF condition $L_gV\geq 0 \Rightarrow L_fV<0$ for a more conservative CLF condition $L_gV\geq 0 \Rightarrow L_fV+L_gV <0$ in the next theorem, which achieves both global stabilization and inverse optimality. Without the CLF restriction, inverse optimality is achievable but the positivity of control is lost. 

\begin{theorem}[Inv. opt. univ. formula---w/ {\em CBF restriction}]
\label{thm:universal_positive}
Under the conditions of Theorem \ref{thm-univ0} but with $V$ satisfying not condition \eqref{eq-CLF-true} but the stronger condition
\begin{equation}
\label{eq-CLF-cond}
a(\xi) \geq 0
\;\Longrightarrow\;
b(\xi) < 0,
\qquad \forall\, \xi\in (0,\infty)\setminus {\bf 1}\,,
\end{equation}
where
\begin{subequations}
\label{eq-a-and-b}
\begin{eqnarray}
a(\xi) &:=& L_{f+g} V(\xi) 
\\
b(\xi) &:=&  L_gV(\xi)\,,
\end{eqnarray}
\end{subequations}
the universal feedback
\begin{equation}
\omega=\omega_\kappa(\xi)=1-
\kappa(\xi) \, 
\max\{0, -{\rm sgn}\{b(\xi)\} \}
\label{eq:control_clean_formal}
\end{equation}
with $\omega_\kappa({\bf 1})=1$ and 
\begin{equation}
\label{eq-kappa-univ}
\kappa(\xi) = \dfrac{a(\xi)+\sqrt{a^2(\xi)+b^2(\xi)}}{b(\xi)}\,, \quad b(\xi)\neq 0
\end{equation}
guarantees that $\dot V(\xi)<0$ for all $\xi\neq{\bf 1}$ and minimizes
\begin{eqnarray}
\label{eq-J-univ}
J&=& \int_0^\infty\frac{b^2(\xi)}{a(\xi)+\sqrt{a^2(\xi) + b^2(\xi)}}
\nonumber\\ && \times 
\left( 1+ \frac{ (\omega-1)^2}{\max\{0, -{\rm sgn}\{b(\xi)\}\}}\right) dt\,
\end{eqnarray}
over all non-negative stabilizing feedback laws, where, when $b(\xi)\ge 0$, the fraction $\frac{ (\omega-1)^2}{\max\{0, -{\rm sgn}\{b(\xi)\}\}}$ is interpreted as $+\infty$ when  $\omega\neq 1$ and as $0$ when $\omega=1$. 
Additionally, under \eqref{eq-CLF-cond}, $\omega$ is continuous on $(0,\infty)^n$.
\end{theorem}

\begin{proof}
Fix $\xi\neq{\bf 1}$. If $b(\xi)\ge0$, then $\omega(\xi)=1$ by \eqref{eq:control_clean_formal} and hence $\dot V(\xi)=a(\xi)<0$ by the CLF assumption. If $b(\xi)<0$, denoting  $s=\sqrt{a^2+b^2}$, 
the control $\omega(\xi)=1+\frac{a+s}{-b}$ yields $\dot V(\xi)=a+b(\omega(\xi)-1)=a+b\frac{a+s}{-b}=a-(a+s)=-s$, i.e., 
\begin{equation}
\label{eq-Vdot-univ}
\dot V =-\sqrt{a^2+b^2}<0\,, \quad \forall \xi\neq\xi_{\rm e}\,.
\end{equation}
The proof of inverse optimality mimics the proof of Theorem 3.2 in \cite{661589}.
For proving continuity, note first that \eqref{eq-CLF-cond} is equivalent to the statement that $a(\xi_0)<0$ whenever $b(\xi_0)=0$ and $\xi_0\neq{\bf 1}$. 
Then, for such $\xi_0$ with $b(\xi_0)=0$ and $a(\xi_0)<0$, take $\xi_k\to\xi_0$ with $b(\xi_k)<0$ and denote $a_k=a(\xi_k)$, $b_k=b(\xi_k)$. For large $k$, $a_k<0$ and the identity $a_k+\sqrt{a_k^2+b_k^2}=\frac{b_k^2}{\sqrt{a_k^2+b_k^2}+|a_k|}$ yields $|\omega(\xi_k)-1|=\frac{a_k+\sqrt{a_k^2+b_k^2}}{-b_k}=\frac{-b_k}{\sqrt{a_k^2+b_k^2}+|a_k|}\leq \frac{|b_k|}{2|a_k|}\to0=\omega(\xi_0)-1$. 
\end{proof}

Theorem~\ref{thm:universal_positive} provides a universal feedback 
tailored to the state equilibrium at ${\bf 1}$ and input equilibrium at $1$, so that the deviation variable is $\omega-1$ and the shifted drift term is incorporated into the definition \eqref{eq-a-and-b}. 
The resulting feedback enforces \eqref{eq-Vdot-univ}
whenever $L_gV(\xi)<0$, while reducing to the equilibrium value $\omega=1$ when $L_gV(\xi)\ge 0$, a region in which the CLF condition guarantees $\dot V = L_fV(\xi)<-L_gV \leq 0$.




Our universal formula \eqref{eq-kappa-univ} looks nearly identical to Sontag's classical formula for input-unconstrained stabilization on $\mathbb{R}^n$ \cite{sontag1989universal}, the most obvious difference being our $b^2$ under the square root, and Sontag's $b^4$ under the square root. Sontag’s $b^4$ ensures continuity at $b=0$ in the unconstrained setting, whereas in the present positive-control setting  $b^2$ reflects the one-sided feasibility  rather than a removable singularity. The main difference, apart from the obvious fact that our $a(\xi)$ differs from the one in Sontag's problem by $+L_gV$, is that our set for $(L_fV+L_gV,L_gV)$ {\em excludes} the closed first quadrant $L_gV \geq 0, L_fV + L_gV \geq 0$, whereas Sontag's set excludes only the negative horizontal line $L_gV=0, L_fV(+L_gV)>0$, since the control can act both positively and negatively. Our ``projection operation'' $\min\{0, {\rm sgn}\{L_gV(\xi)\} \}$ in \eqref{eq:control_clean_formal},  deactivates the control formula $\kappa$ in the second quadrant of the $(a,b)$ plane, namely, for all $\xi$ where $L_gV(\xi)>0$.  

The inverse optimal result \eqref{eq-J-univ}, for our universal formula \eqref{eq:control_clean_formal}, \eqref{eq-kappa-univ}, is meaningful, since $\frac{b^2(\xi)}{a(\xi)+\sqrt{a^2(\xi) + b^2(\xi)}}$ is positive definite on the positive orthant (due to the CLF condition) and $\frac{ (\omega-1)^2}{\max\{0, -{\rm sgn}\{b(\xi)\}\}}$, the weight on the deviation of $\omega$ from its equilibrium value $1$, is infinite when $L_gV\geq 0$, which by \eqref{eq-CLF-cond} implies that $L_fV\leq L_fV + L_gV<0$, namely, that $\dot V$ is negative and control effort other than the equilibrium-keeping control $\omega =1 $ would be wasted.

Remark \ref{rem-S+} indicated the inapplicability of the  Sontag-Lin  formula \cite[(22), (19), (16]{LinSontag1995} to the predator-prey model with the CLF \eqref{V-bkst}. For the same reason, our inverse-optimal  formula \eqref{eq:control_clean_formal}, \eqref{eq-kappa-univ} is not applicable either. However, as pointed out in Remark \ref{rem-pred-univ}, our formula \eqref{eq-univ-basic} is applicable. 

\begin{example}{\bf (Alternative to Lin-Sontag \cite[(49)]{LinSontag1995}
\label{ex-LinSon1})}
\rm Theorem \ref{thm:universal_positive} is written for systems 
in the positive orthant, equilibrium at ${\bf 1}$, positive control, and control equilibrium at $1$. It is also applicable to systems in $\mathbb{R}^n$, equilibrium at zero for zero control, and control non-negative. This is the subject of \cite{LinSontag1995}---treated there with positive control. For the  example $\dot x = x^2 - u$ with $u>0$ a feedback law \cite[(49)]{LinSontag1995} is given by
\begin{eqnarray}
    u&=& \frac{3 \pi }{4(2\pi - \theta)} \Bigg[ 
    \theta-\frac{\pi}{2}
    + \frac{2\theta - 3\pi}{\pi}\arctan\frac{2\theta-3\pi}{|x|\sqrt{x^2+1}}
    \nonumber\\
    && -2\arctan\frac{\pi}{|x|\sqrt{x^2+1}}\Bigg]
    \\
    \theta &=& -\arctan\frac{1}{x^2} + \pi +\begin{cases}
        \pi,& x>0\\
        0, & x<0\,.
    \end{cases}
\end{eqnarray}
Our  Theorem \ref{thm:universal_positive}
yields the simple controller
\begin{equation}
u^*(x)
=
\begin{cases}
x^2+\sqrt{x^4+1}, & x>0\\
0, & x\le0
\end{cases}
\end{equation}
which is not only stabilizing with  
\begin{equation}
\dot V(x)
=
-|x|
\begin{cases}
\sqrt{x^4+1}, & x>0\\
x^2, & x\le0
\end{cases}
\end{equation}
but also optimal with state cost 
$q(x)
=
\frac{|x|}{x|x|+\sqrt{x^4+1}}$ and weight 
\begin{equation}
r(x)
=
\begin{cases}
q(x), & x>0\\
+\infty, & x\le0
\end{cases}
\end{equation}
on the control cost $u^2$.
\hfill $\triangle$
\end{example}

\section{General Expander-Parametrized Positive Inverse-Optimal Stabilizers} 
\label{secX}

The universal formulae in Section \ref{secX-} are 
are only examples of feedback laws, where the user still has to first design a CLF-affirming feedback. 
The predator-prey benchmark 
provides a roadmap for  a generalized redesign framework, introduced in the next theorem.




\begin{theorem}[Expander-based inv. optimal {\em redesign}]
\label{thm3}
Let 
\begin{itemize}
\item 
$\omega_0:(0,\infty)^n \to (0,\infty)$ satisfy $\omega_0(\mathbf{1}) = 1$ and the strict CLF inequality
\begin{equation}
\label{eq-q-def-xi}
q(\xi) := -L_f V(\xi) - L_g V(\xi)\,\omega_0(\xi) > 0,
\qquad \forall\, \xi \neq \mathbf{1}.
\end{equation}
\item 
$\Sigma:(0,\infty)\to(0,\infty)$ be $C^1$, strictly increasing, onto, normalized by $\Sigma(1)=1$, and satisfy the sign-alignment condition
\begin{equation}
\label{def-sign-condition}
\fbox{$\operatorname{sign}\!\Big(\Sigma\big(\omega_0(\xi)\big)-1\Big)
=-\operatorname{sign}\!\big(L_g V(\xi)\big),
\quad \forall\, \xi \neq \mathbf{1}$}\,.
\end{equation}
\end{itemize}
Then, for every such $\Sigma$, the feedback
\begin{equation}
\label{eq-expander-redesign}
\omega^*(\xi) := \Sigma\big(\omega_0(\xi)\big)
\end{equation}
minimizes, over all globally asymptotically stabilizing controllers, the infinite-horizon functional
\begin{equation}
J(\omega)
=
\int_0^\infty
\Big(
q(\xi(t))
+
r(\xi(t))\,\Psi(\omega(t))
\Big)\,dt
\end{equation}
along the solutions of \eqref{eq-xi-sys}, where $\Psi:(0,\infty)\to\mathbb{R}_{\ge 0}$ is the strictly convex penalty (unique up to a positive scale factor) defined implicitly by
\begin{equation}
\label{eq-Psi-Sigma-def}
\frac{\Psi(s)}{\Psi'(s)} = s-\Sigma^{-1}(s) ,
\qquad s \neq 1,
\qquad \Psi(1) = 0,
\end{equation}
the weight $r(\xi)$ is strictly positive on $(0,\infty)^n$ 
(with a positive extended-real limit at $\xi=\mathbf{1}$), and the product $r(\xi)\Psi(\omega)$, equivalently expressed as
\begin{eqnarray}
\label{eq-67}
r(\xi)\Psi(\omega)
&=&
L_g V(\xi)\Big(\omega_0(\xi)-\Sigma(\omega_0(\xi))\Big)
\nonumber\\
&& \times\exp\!\left(
\int_{\omega_0(\xi)}^{\Sigma^{-1}(\omega)}
\frac{\Sigma'(\sigma)}{\Sigma(\sigma)-\sigma}\,d\sigma
\right),
\end{eqnarray}
 is strictly convex in $\omega$ for all $\xi \in (0,\infty)^n$ and strictly positive for all $\omega \neq 1$. Moreover, $V$ is the value function, 
\begin{equation}
J^*=V(\xi_0)
\end{equation}
and the optimal closed loop 
is globally asymptotically stable on  $(0,\infty)^n$ at $\xi=\mathbf{1}$.
\end{theorem}

\begin{proof}
From \eqref{eq-xi-sys}, \eqref{eq-q-def-xi},
\begin{equation}
\dot V
=
L_f V + L_g V\,\omega
=
- q(\xi) + L_g V(\xi)\,(\omega - \omega_0(\xi)).
\end{equation}
Add and subtract $\omega^*$:
\begin{equation}
\dot V
=
- q(\xi)
+
L_g V(\xi)\,(\omega - \omega^*)
+
L_g V(\xi)\,(\omega^* - \omega_0).
\end{equation}
From the defining identity \eqref{eq-Psi-Sigma-def} 
evaluated at $s=\omega^*$ and using $\Sigma^{-1}(\omega^*)=\omega_0$, we obtain
\begin{equation}
\label{eq-72}
\Psi(\omega^*)
=
(\omega^* - \omega_0)\,\Psi'(\omega^*).
\end{equation}
Define
\begin{equation}
\label{eq-73}
r(\xi)
=
- \frac{L_g V(\xi)}{\Psi'(\omega^*(\xi))}.
\end{equation}
Using \eqref{def-sign-condition}, $\omega^*=\Sigma(\omega_0)$, and $\operatorname{sign}\Psi'(s)=\operatorname{sign}(s-1)$, we get $L_gV(\xi)\Psi'(\omega^*(\xi))<0$ for $\xi\neq\mathbf1$, hence $r(\xi)=-L_gV(\xi)/\Psi'(\omega^*(\xi))>0$ on $(0,\infty)^n\setminus\{\mathbf1\}$.
Moreover, since $r(\xi)>0$ for every $\xi\neq\mathbf 1$ and the numerator and denominator in \eqref{eq-73} have the same sign in a neighborhood of $\mathbf 1$, it follows that any extended-real limit value of $r(\xi)$ as $\xi\to\mathbf 1$ is positive (possibly $+\infty$), and therefore $r(\mathbf 1)>0$ in the extended-real sense.
Next, from \eqref{eq-72}, \eqref{eq-73},
\begin{equation}
\label{eq-74}
L_g V(\xi)\,(\omega^* - \omega_0)
=
- r(\xi)\,\Psi(\omega^*(\xi)).
\end{equation}
Substituting,
\begin{equation}
\dot V
=
- q(\xi)
+
L_g V(\xi)\,(\omega - \omega^*)
-
r(\xi)\,\Psi(\omega^*(\xi)).
\end{equation}
The construction \eqref{eq-Psi-Sigma-def}, equivalent to \eqref{eq-Theta-to-Psi} with $\Theta=\Sigma^{-1}$, yields a $C^1$ strictly convex $\Psi$ with unique minimizer at $1$, hence $\operatorname{sign}\Psi'(s)=\operatorname{sign}(s-1)$ for all $s\neq 1$.
By convexity of $\Psi$,
\begin{equation}
\Psi(\omega)
\ge
\Psi(\omega^*)
+
\Psi'(\omega^*)\,(\omega - \omega^*).
\end{equation}
Multiplying by $r(\xi)$ and using the definition of $r$,
\begin{equation}
L_g V(\xi)\,(\omega - \omega^*)
\ge
r(\xi)\,\Psi(\omega^*)
-
r(\xi)\,\Psi(\omega).
\end{equation}
Hence
\begin{equation}
\dot V
\le
- q(\xi)
-
r(\xi)\,\Psi(\omega).
\end{equation}
For $\xi\neq\mathbf1$, equality holds if and only if $\omega=\omega^*(\xi)$. Integrating over $[0,\infty)$ gives
\begin{align}
&\int_0^\infty\big(q(\xi(t))+r(\xi(t))\,\Psi(\omega(t))\big)\,dt
\\
&\ge V(\xi(0))-\lim\inf_{t\rightarrow\infty} V(\xi(t)) \ge V(\xi(0)),
\end{align}
with equality precisely for $\omega=\omega^*$. 
To derive \eqref{eq-67}, 
let $s=\Sigma(\sigma)$. Since $\Sigma^{-1}(\Sigma(\sigma))=\sigma$, this yields
${\Psi'(\Sigma(\sigma))}/{\Psi(\Sigma(\sigma))} = {1}/{(\Sigma(\sigma)-\sigma)}$
from \eqref{eq-Psi-Sigma-def}.
Differentiating $\ln \Psi(\Sigma(\sigma))$ with respect to $\sigma$ gives
\begin{equation}
\frac{d}{d\sigma}\ln \Psi(\Sigma(\sigma))
=
\frac{\Psi'(\Sigma(\sigma))}{\Psi(\Sigma(\sigma))}\,\Sigma'(\sigma)
=
\frac{\Sigma'(\sigma)}{\Sigma(\sigma)-\sigma}.
\end{equation}
Integrating from $\sigma=\omega_0(\xi)$ to $\sigma=\Sigma^{-1}(\omega)$ and using $\Sigma(\omega_0(\xi))=\omega^*(\xi)$,
\begin{equation}
\label{eq-82}
\Psi(\omega)
=
\Psi(\omega^*(\xi))\,
\exp\!\left(
\int_{\omega_0(\xi)}^{\Sigma^{-1}(\omega)}
\frac{\Sigma'(\sigma)}{\Sigma(\sigma)-\sigma}\,d\sigma
\right).
\end{equation}
Multiplying \eqref{eq-82} by $r(\xi)$ and substituting \eqref{eq-74} yields \eqref{eq-67}.
Under $\omega=\omega^*$,
\begin{equation}
\dot V
=
- q(\xi)
-
r(\xi)\,\Psi(\omega^*(\xi))
<
0
\quad
(\xi\neq \mathbf{1}),
\end{equation}
so $\mathbf{1}$ is globally asymp. stable and $V$ is the value function.
\end{proof}

\vspace{-3mm}\paragraph*{Design pathway.} As difficult s designing simultaneously a strict CLF and a globally stabilizing controller to satisfy \eqref{eq-q-def-xi}, even more difficult is ensuring that such a $(V,\omega_0)$ pair lends itself to satisfying condition \eqref{def-sign-condition} of sign-compatibility for some admissible (increasing, invertible, expansive) $\Sigma$. 

To get a taste of what it takes to satisfy  \eqref{def-sign-condition}, we remind the reader that, for the predator-prey model, it becomes  
$X>Y \;\Longleftrightarrow\; \Sigma\!\left({Y}/{X}\right)<1$ and $
X<Y \;\Longleftrightarrow\; \Sigma\!\left({Y}/{X}\right)>1$.
While seemingly simple, meeting the sign-compatibility is anything but easy: thanks to a carefully constructed pair comprising  the backstepping-design Volterra-Lyapunov $V$ with backstepping designed nominal feedback $\omega_0 = Y/X$ . 

Condition \eqref{def-sign-condition} is best understood as the positive-input counterpart of the classical 
negation of $L_gV$ in \eqref{eq-xdot=f+gu}. 
For LTI systems, this ``negative of $L_gV$'' form has been known since Kalman answered the question "When is a linear control system optimal?" \cite{10.1115/1.3653115} with his celebrated {\em inverse-optimal} feedback form $u= -R^{-1} B^{-1} P x$, where $P$ is a solution of not a Riccati but a Lyapunov algebraic equation. 

\vspace{-3mm}\paragraph*{Reciprocal symmetry.} The function 
$(\ell\gamma)' \circ \left({\ell\gamma}/{\mathrm{Id}}\right)^{-1}$ in \eqref{eq-inv-opt-class-u*-altern2} is the analog
of the expander $\Sigma$, and $
\left({\ell\gamma}/{\mathrm{Id}}\right)
\circ
\big((\ell\gamma)'\big)^{-1}$ the analog of the contractor $\Theta$. The difference that arises in positive systems is that, while $(\ell\gamma)' \circ \left({\ell\gamma}/{\mathrm{Id}}\right)^{-1}$ is odd in a real variable (nominal sign-unconstrained control),  expander $\Sigma$ is, in general, not symmetric relative to $s=1$ on $(0,\infty)$. 

However, `unconventionally symmetric' $\Sigma$'s  exist. A {\em ``reciprocally symmetric''} $\Sigma(\sigma)={1}/{\Sigma\!\big(\sigma/\Sigma(\sigma)\big)},
\ \forall\,\sigma>0$ results in a control cost $\Psi$ that is not only strictly convex but itself satisfies the reciprocal symmetry of the form $\Upsilon(s) + \Upsilon(1/s)=1$, where $\Upsilon(s) = \dfrac{\Psi(s)}{(s-1)\Psi'(s)}$. The fact is proven with $\sigma = \Theta(s)$. This---reciprocal---form of symmetry is interesting because it holds for $\Sigma = \Theta^{-1}$ in \eqref{eq-def-Theta-Volterra} and $\Psi=\Omega $ defined by \eqref{eq-Volterra-Lyapunov}, as well as for $\Sigma(s)=s^2$ and the corresponding $\Psi(s)=(\sqrt{s}-1)^2$. 

\vspace{-3mm}\paragraph*{Expander with an ordinary {\em odd} symmetry in the universal formula.} If one takes the universal formula \eqref{eq:control_clean_formal}, \eqref{eq-kappa-univ} divided by {\em two}, and appropriately shifted by $1$, namely, 
\begin{equation}
\omega_0 = \omega_{\kappa/2}(\xi) := 1 + \frac{1}{2}(\omega_\kappa(\xi)-1)\,,
\label{eq:control_clean_formal-1/2}
\end{equation}
this half-universal feedback can be verified to guarantee
\begin{equation}
\dot V = -\frac12\sqrt{a^2+b^2}<0\,, \quad \forall \xi\neq\xi_{\rm e}\,,
\end{equation}
This is a signature of inverse optimality of $\omega_\kappa$, 
established in Theorem \ref{thm:universal_positive}. The question then is how to interpret this contraction relation between $\omega_\kappa$ and $\omega_{\kappa/2}$ using Theorem \ref{thm3}, and specifically with its expander redesign formula \eqref{eq-expander-redesign}. The answer is simple---by combining \eqref{eq:control_clean_formal-1/2}, 
\eqref{eq-expander-redesign}, one gets
\begin{equation}
\omega_\kappa = 1 + 2 (\omega_{\kappa/2} -1)\,,
\end{equation}
namely, the redesign formula \eqref{eq-expander-redesign} holds with the {\em expander}
\begin{equation}
\Sigma(s)=2s\,,
\end{equation} 
a $\mathbb{R}\rightarrow \mathbb{R}$ function with odd symmetry around $s=0$. It is largely the anticipation of this feature of universal control---nor merely producing a universal formula but identifying it as a particular case of the inverse optimal redesign in Theorem \ref{thm3}---that has driven us to develop Section \ref{secX-}.

\section{Exploring the Parametrizing Families of Contractor Functions and Control Penalties}
\label{secVIII}

Theorems \ref{sec-inv-opt} and \ref{thm3} present a broad parametrization of optimal controllers $U^*= Y 
\Sigma(Y/X)$, resp. $\omega*=\Sigma(\omega_0)$,  and control cost functions $\Psi(U/Y)$, resp. $\Psi(\omega)$, in terms of contractors $\Theta$. This generality deserves a bit of illustration.

We carefully choose three examples of contractor functions with different behaviors as $s\rightarrow 0$ and $s\rightarrow\infty$:
\begin{subequations}
\label{eq-Theta0i}
\begin{eqnarray}
\label{eq-Theta01}
\Theta_1(s)&=&\frac{s\ln s}{ s-1}
\\
\label{eq-Theta02}
\Theta_2(s)&=&\sqrt{s}
\\
\label{eq-Theta03}
\Theta_3(s)&=&\frac{2s-\sqrt{s^2+1}+1}{3-\sqrt{2}}\,.
\end{eqnarray}
\end{subequations}
The asymptotic behaviors of $\Theta_1, \Theta_2, \Theta_3$ are
\begin{subequations}
\label{eq-Theta0i+}
\begin{eqnarray}
\label{eq-Theta01+}
\Theta_1\sim s|\ln s| \quad (s\to 0^+),
&&
\Theta_1\sim \ln s \quad (s\to\infty)\qquad
\\
\label{eq-Theta02+}
\Theta_2\sim s^{1/2} \quad (s\to 0^+),
&&
\Theta_2\sim s^{1/2} \quad (s\to\infty)
\\
\label{eq-Theta03+}
\Theta_3\sim c\,s \quad (s\to 0^+),
&&
\Theta_3\sim \kappa\, s \quad (s\to\infty)
\\
c=\frac{2}{3-\sqrt{2}},
&&
\kappa=\frac{1}{2(3-\sqrt{2})}
\end{eqnarray}
\end{subequations}
namely, $\Theta_3, \Theta_1, \Theta_2$ are, respectively, the least, medium, and most aggressively growing near $s=0$, whereas $\Theta_1, \Theta_2, \Theta_3$ are, respectively, the least, medium, and most aggressively growing for large $s$. These functions' trends act in an inverse fashion when employed in the feedback \eqref{eq-optimal-U*}
through \eqref{eq-def-Sigma}. 

\begin{figure}[t]
    \centering
    \includegraphics[width=0.8\linewidth]{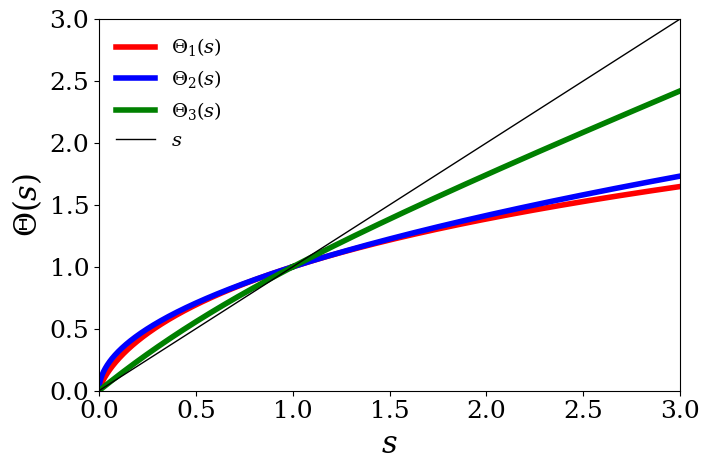}
    \caption{The ``contractor'' nonlinearities $\Theta_i$ defined in \eqref{eq-Theta0i}, with their growth trends quantified in \eqref{eq-Theta0i+}.}
    \label{fig:Theta}
\end{figure}

\begin{figure}[t]
    \centering
    \includegraphics[width=0.8\linewidth]{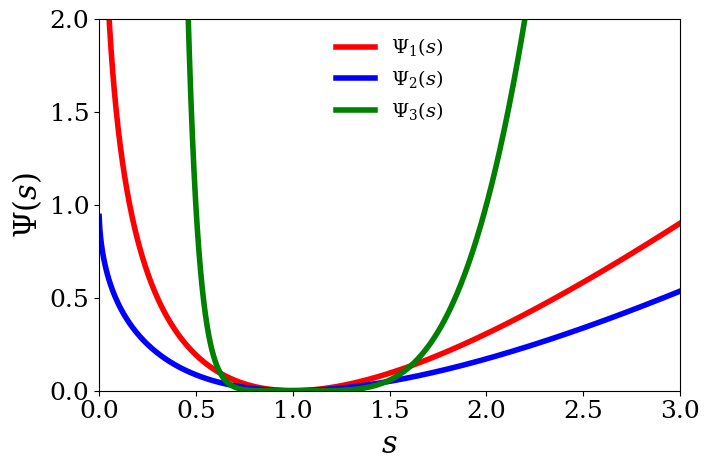}
    \caption{The control cost functions $\Psi_i$, obtained from the respective $\Theta_i$, using \eqref{eq-Theta-to-Psi}. Both $\Psi_1$ and $\Psi_3$ have barrier behavior at $s=0$, but $\Psi_2$ does not. Of the three, $\Psi_1$ and $\Psi_1$ have the least aggressive growth for large $s$ --- linear.}
    \label{fig:Psi}
\end{figure}

The corresponding $\Psi$-functions, defined by \eqref{eq-Psidef}, are
\begin{subequations}
\label{eq-Psi0}
\begin{eqnarray}
\label{eq-Psi01}
    \Psi_1(s)&=& s-1-\ln s,
\\
\label{eq-Psi02}
\Psi_2(s)&=&(\sqrt{s}-1)^2
\\
\nonumber
\label{eq-Psi03}
\Psi_3(s) &=& \mbox{see \eqref{eq-Theta-to-Psi}, not available in closed form,}
\\
\nonumber&& \mbox{but see \eqref{eq-Psi0+}}
\end{eqnarray}
\end{subequations}    
Their asymptotic behaviors are, for $\Psi_1$,
\begin{subequations}
\begin{eqnarray}
\Psi_1(0^+)&=&+\infty,
\qquad
\Psi_1(s)\sim |\ln s| \ (s\to 0^+),
\\ 
\Psi_1(s)&\sim& s \ (s\to\infty),
\\ 
\Psi_1(s)&\sim& \tfrac12 (s-1)^2 \ (s\to 1)
\end{eqnarray}
\end{subequations}
for $\Psi_2$,
\begin{subequations}
\begin{eqnarray}
\Psi_2(0^+)&=&1,
\qquad 
\Psi_2(s)=1-2\sqrt{s}+s \ (s\to 0^+),
\\ 
\Psi_2(s)&\sim& s \ (s\to\infty),
\\ 
\Psi_2(s)&\sim& \tfrac14 (s-1)^2 \ (s\to 1)
\end{eqnarray}
\end{subequations}
and for $\Psi_3$,
\begin{subequations}
\label{eq-Psi0+}
\begin{eqnarray}
\Psi_3(0^+)&=&+\infty,
\qquad 
\Psi_3(s)\sim C\, s^{-(1+2\sqrt{2})} \ (s\to 0^+),
\\ 
\Psi_3(s)&\sim& C\, s^{(22+2\sqrt{2})/17} \ (s\to\infty),
\\ 
\Psi_3(s)&\sim& C\, |s-1|^{\,4+\sqrt{2}} \ (s\to 1).
\end{eqnarray}
\end{subequations}
As evident in the equations above, as well as shown in Figure \ref{fig:Psi}, both $\Psi_1$ and $\Psi_3$ have barrier behavior at $s=0$, but $\Psi_2$ does not. Of the three, $\Psi_1$ and $\Psi_2$ has the least aggressive growth for large $s$ --- linear --- while $\Psi_3$ is superlinear. Around the minimum, $\Psi_1$ and $\Psi_2$ are quadratic in $s-1$, whereas $\Psi_3$ is more than twice as flat. 

In summary, the contractor $\Theta$, 
spans optimal controllers that are cautious or aggressive, near or far from  setpoint, with meaningful costs on control.

\section{Constructive General Inverse Optimal Stabilizers}
\label{secX+}

In Section \ref{secX} we parametrized the inverse optimal stabilizing controllers using an expander function $\Sigma$. That is a very general framework, useful when the $L_gV$-directionality condition \eqref{def-sign-condition} is satisfied by the nominal controller $\omega_0$ and the expander, but 
a fully constructive framework is preferable, in which a controller is {\em directly designed to be stabilizing and inverse optimal}. Those two properties are provided by the single universal controllers \eqref{eq:control_clean_formal}, \eqref{eq-kappa-univ}, however, one wants a broader performance-shaping catalog of controllers. 

The constructive framework we are after in this section is equivalent to designing \eqref{eq-inv-opt-class-u*}, 
simultaneously  over $(\gamma,r)$. We provide a design with simultaneous choices of $(\Theta,r)$. The contractor $\Theta$ is chosen first, followed by a choice of the control weight $r$, obtaining in the end both an inverse optimally stabilizing controller, and a controller that is merely stabilizing and analogous to the half-universal formula $\omega_{\kappa/2}$. 

Consider system \eqref{eq-xi-sys} and the associated
\begin{equation}
\label{eq-Vdot-r-shape}
\dot V = a + b (\omega-1)\,, 
\end{equation}
where $a(\xi),b(\xi)$ are defined in \eqref{eq-a-and-b}, and denote by  $S \subset \mathbb{R}^2$  the set where $b \ge 0 \Rightarrow a < 0$. Although the designs are expressed in terms of Lie-derivatives  $(a,b)$, the controllers are feedbacks of the state $\xi\in(0,\infty)^n$. 
Compared to Theorem \ref{thm3}, the design in the next theorem is 
more restrictive, 
with $L_gV\geq 0 \Rightarrow L_fV+L_gV <0$, as in the universal  Theorem \ref{thm:universal_positive}, 
to make optimality possible with positive control. 

Theorem \ref{thm-5} offers far more than Theorem \ref{thm:universal_positive}: controller family parametrized with contractor  $\Theta$. 

\begin{theorem}[Inv. opt. {\em direct} design---w CBF restriction]
\label{thm-5}
Let 
\begin{itemize}
\item the strictly increasing function $\Theta:(0,\infty)\to(0,\infty)$ satisfy the conditions in Lemma \ref{lem-Theta-Psi};  
\item  the strictly convex $\Psi:(0,\infty)\to(0,\infty)$ be mapped from $\Theta$ in accordance with \eqref{eq-Psidef} and with the properties established in Lemma \ref{lem-Theta-Psi}; 
\item the strictly increasing $\Gamma:(0,\infty) \rightarrow (1,\infty)$ and $\Gamma^{-1}:(1,\infty)\rightarrow(0,1)$ denote 
\begin{subequations}
\begin{eqnarray}
\Gamma &=& \Theta \circ (\Psi')^{-1}\\
\Gamma^{-1} &=& \Psi' \circ \Theta^{-1}\,, 
\end{eqnarray}
\end{subequations}
where, under Lemma \ref{lem-Theta-Psi}, 
$\Gamma(0)=1$, 
and $\Gamma^{-1}(1)=0$; 
\item  the user choose the function $r=r(a,b)\in(0,\infty]$ such that $r(a,b)=+\infty$ for $b \ge 0$, and for $b<0$,
\begin{equation}
r(a,b) > -b \quad \mbox{when} \ \ a\leq 0
\label{eq:r_positive_clean}
\end{equation}
and
\begin{equation}
r(a,b)\in \left(-b, - \frac{b}{\Gamma^{-1}(1-a/b)}\right)  \quad \mbox{when}\ \  a> 0\,.
\label{eq:r_condition_clean}
\end{equation}
\end{itemize}
Then the feedback
\begin{equation}
\omega^*(a,b)=(\Psi')^{-1}\left(-\frac{b}{r}\right)\,,
\label{eq:omega_star_clean}
\end{equation}
with the limiting value $\omega^*(a,b)=1$ when $r=+\infty$, minimizes
\begin{equation}
J(\omega)=\int_0^\infty \Big[q(a,b)+r(a,b)
\Psi(\omega)
\Big]\,dt
\label{eq:J_clean}
\end{equation}
along the solutions of \eqref{eq-xi-sys}, where
\begin{equation}
q(a,b)=-a+b
-b\, \Gamma\left(-\frac{b}{r}\right)\,,
\label{eq:q_clean}
\end{equation}
with limiting value $q(a,b)=-a$ when $r=+\infty$, has the property that 
\begin{equation}
q(a,b)>0 \quad \text{for all } (a,b)\in S\setminus\{(0,0)\}\,
\label{eq:q_posdef_clean}
\end{equation}
and $q(0,0)=0$. Furthermore, both the feedback \eqref{eq:omega_star_clean} and 
\begin{equation}
\omega_0(a,b)=
\Gamma\left(-\frac{b}{r}\right)\,,
\label{eq:omega_star_clean0}
\end{equation}
with the limiting value $\omega_0(a,b)=1$ when $r=+\infty$, are stabilizing at $a=b=0$ on $S$ with respect to the CLF $V$. 
\end{theorem}

\begin{proof}
For $(a,b)\in S$, consider the function
\begin{equation}
h(\omega)=b(\omega-1)+r
\Psi(\omega),
\label{eq:h_def}
\end{equation}
where $r=r(a,b)$ and $\omega>0$. Since $\Psi(\omega)
$ is strictly convex on $(0,\infty)$, so is $h(\omega)$ whenever $r<\infty$. Its derivative is
$h'(\omega)=b+r\Psi'(\omega)$, 
and setting it 
equal to zero yields
\begin{equation}
b+r\Psi'(\omega)
=0,
\label{eq:stationarity}
\end{equation}
whose unique solution is \eqref{eq:omega_star_clean}. 
In the case $b\ge 0$, the theorem sets $r=+\infty$, and the penalty term enforces $\omega=1=\omega^*$ as the unique minimizer. Hence $\omega^*$ uniquely minimizes $h$ in all cases.
Substituting \eqref{eq:omega_star_clean} gives
\begin{equation}
b(\omega^*-1)= b \left[(\Psi')^{-1}\left(-\frac{b}{r}\right)-1\right],
\label{eq:identity1}
\end{equation}
\begin{equation}
r\Psi(\omega^*)=-b\,(\omega^*-\Theta(\omega^*)),
\label{eq:identity2}
\end{equation}
and therefore
\begin{equation}
a+b(\omega^*-1)+r\Psi(\omega^*)
=-q(a,b).
\label{eq:hjb_identity}
\end{equation}
Let $\omega(t)>0$ be arbitrary and consider \eqref{eq-Vdot-r-shape}. 
By convexity of $\Psi$,
\begin{equation}
\Psi(\omega)\ge \Psi(\omega^*)+\Psi'(\omega^*)(\omega-\omega^*).
\label{eq:convexity}
\end{equation}
Multiplying \eqref{eq:convexity} by $r$ and using the stationarity condition \eqref{eq:stationarity} yields
\begin{equation}
b(\omega-\omega^*)\ge r\Psi(\omega^*)-r\Psi(\omega).
\label{eq:key_inequality}
\end{equation}
Combining \eqref{eq-Vdot-r-shape}, \eqref{eq:hjb_identity}, and \eqref{eq:key_inequality} gives
\begin{equation}
\dot V \ge -q(a,b)-r\Psi(\omega).
\label{eq:Vdot_bound}
\end{equation}
Integrating \eqref{eq:Vdot_bound} over $[0,T]$,
\begin{equation}
\int_0^T\!\big(q(a,b)+r\Psi(\omega)\big)\,dt
\ge V(0)-V(T).
\label{eq:integrated}
\end{equation}
Since $V\ge 0$, letting $T\to\infty$ yields $J(\omega)\ge V(0)$. For $\omega\equiv\omega^*$, all inequalities hold with equality, so $J(\omega^*)=V(0)$, and therefore $\omega^*$ minimizes $J$. 
The fact that $q(a,b)$ is positive definite is proven next.
On the branch $b<0$, condition \eqref{eq:r_condition_clean} implies
$-\frac{b}{r} > \Gamma^{-1}(1-a/b)$.
Since $\Gamma$ is strictly increasing on (0,1), it follows that
$\Gamma\!\left(-\frac{b}{r}\right) > 1-a/b$.
Multiplying both sides by $b<0$ reverses the inequality, yielding
$-b\,\Gamma\!\left(-\frac{b}{r}\right) < a-b$.
Therefore,
$q(a,b) = -a + b - b\,\Gamma\!\left(-\frac{b}{r}\right) > 0$.
On the branch $b\ge 0$, the theorem assigns $r=+\infty$, hence $\omega^*=1$ and $q(a,b)=-a$, and since $(a,b)\in S$ implies $b\ge 0\Rightarrow a<0$, it follows that $q(a,b)=-a>0$. 
Finally, 
$q(0,0)=0$ at equilibrium by definition. 
Now we turn to proving stabilization. 
Fix $(a,b)\in S\setminus\{(0,0)\}$. If $b\ge 0$, then one has $r=+\infty$ and
$\dot V=a+b(\omega-1)=a<0$,
since $b\ge 0\Rightarrow a<0$ on $S$. If $b<0$, then $r>-b$ and hence $z:=b/r\in(-1,0)$, and we have 
$\omega=\omega_0=\Gamma(-z)=\Theta(\omega^*)$;
 thus
\begin{equation}
\dot V=a+b(\omega_0-1)
=a+ b\Gamma(-z)-b
=-q(a,b)<0.
\end{equation}
For $\omega^*$, 
with \eqref{eq:hjb_identity}, 
$\dot V=a+b(\omega^*-1)= - q(a,b)-r\Psi(\omega^*)$. 
Using \eqref{eq:hjb_identity},
and the facts that $q(a,b)>0$ and $\Psi(\omega^*)
>0$ 
for $\omega^*\neq 1$, we conclude $\dot V<0$ under $\omega^*$ as well.
\end{proof}


Inequalities \eqref{eq:r_positive_clean}, \eqref{eq:r_condition_clean} are not  assumptions. A function $r(a,b)$ that satisfies these inequalities can always be selected. So, controllers \eqref{eq:omega_star_clean} and \eqref{eq:omega_star_clean0} are {\em families} of universal stabilizers, and \eqref{eq:omega_star_clean}, additionally, inverse optimal.

\vspace{-3mm}\paragraph*{Enforcing continuity of feedback.} One can choose $r(a,b)$ for $b<0, a>0$ to make the optimizer $\omega^*(a,b)$ continuous across the boundary $a=0$ in infinitely many ways, including  $r=-{b}/{\sqrt{\Gamma^{-1}(1-a/b)}}$ or $r=-\frac{b\bigl(1+\Gamma^{-1}(1-a/b)\bigr)}{2\,\Gamma^{-1}(1-a/b)}$. We select the latter, namely, $r$ that is midpoint within the interval \eqref{eq:r_condition_clean}. So, for all $(a,b)\in S$, we choose
\begin{equation}
r(a,b)=
\begin{cases}
-b\dfrac{1+\Gamma^{-1}(1-a/b)}{2\Gamma^{-1}(1-a/b)},
& b<0 \mbox{ and } a>0,\\[10pt]
+\infty, & \mbox{else}
\end{cases}
\end{equation}
where $\Gamma^{-1}:(1,\infty)\to(0,1)$ is strictly increasing and satisfies $\Gamma^{-1}(1)=0$. With this choice,
\begin{eqnarray}
\omega^*(a,b)
&=&
(\Psi')^{-1}\!\left(-\frac{b}{r(a,b)}\right)
\nonumber \\
&=&
\begin{cases}
(\Psi')^{-1}\!\Bigg(
\dfrac{2\Gamma^{-1}(1-a/b)}{1+\Gamma^{-1}(1-a/b)}
\Bigg)
, & b<0 \mbox{ and }a>0,\\[6pt]
1, & \mbox{else.}
\end{cases}
\nonumber\\ &&
\end{eqnarray}
Since $\Gamma^{-1}(1-a/b)\to 0$ as $a\downarrow 0$, one has $r(a,b)\to+\infty$ and hence $\omega^*(a,b)\to (\Psi')^{-1}(0)=1$, ensuring continuity across the boundary $a=0$. Moreover, along the line $b=0$ with $a<0$, the condition $a\le 0$ is automatically satisfied, so $r=+\infty$ and $\omega^*=1$. 
This continuity-ensuring choice for $\omega^*$ makes $\omega_0$ also continuous, since $\omega_0=\Theta(\omega^*)$. 

The continuous design is 
a corollary of Thm. \ref{thm-5}.

\begin{corollary}[{\em Continuous explicit} inv. opt. design]
\label{thm-5cc}
With the notation in Theorem \ref{thm-5}, the continuous and positive feedback laws
\begin{subequations}
\label{eq-construct-contin}
\begin{eqnarray}
\label{eq-construct-contin*}
\omega^* &=&
\begin{cases}
(\Psi')^{-1}\!
\Bigg(
\dfrac{2\zeta}{1+\zeta}\Bigg)
, & 
L_gV <0 ,
\ L_{f+g}V>0\\[6pt]
1, & \mbox{else}
\end{cases}
\nonumber\\
&&
\\[5pt]
\label{eq-construct-contin0}
\omega_0 &=&\Theta(\omega^*)
\end{eqnarray}
\end{subequations}
where
\begin{equation}
    \zeta(\xi) = \Gamma^{-1}\left( -\frac{L_fV(\xi)}{L_gV(\xi)}\right) 
\end{equation}
are both globally stabilizing at 
$\xi={\bf 1}$ on $(0,\infty)^n$
and \eqref{eq-construct-contin*} minimizes \eqref{eq:J_clean} with 
\begin{equation}
q = - L_fV - L_gV 
\begin{cases}
\Gamma \Bigg(
\dfrac{2\zeta}{1+\zeta}\Bigg), 
& L_gV <0 , 
\ L_{f+g}V>0 \\[10pt]
1, & \mbox{else}
\end{cases}
\end{equation}
which is positive for $\xi\in(0,\infty)\setminus {\bf 1}$ and the positive 
\begin{equation}
r 
=
\begin{cases}
-L_gV \Bigg(
\dfrac{1+\zeta}{2\zeta}\Bigg)
, & L_gV 
<0 , 
\ L_{f+g}V>0\\[10pt]
+\infty, & \mbox{else}.
\end{cases}
\end{equation}
\end{corollary}

Next we specialize Theorem \ref{thm-5} to the canonical, Volterra $\Psi=\Omega$, which fixes the contractor design choice to $\Theta$ in \eqref{eq-def-Theta-Volterra}, and then designs $r(\xi)$. 

\begin{corollary}[Inv. opt. design w/ {\em Volterra} cost]
\label{thm-5c}
With the function $r=r(a,b)\in(0,\infty]$  chosen such that $r(a,b)=+\infty$ for $b \ge 0$, and for $b<0$,
\begin{equation}
r(a,b) > -b,
\label{cor-eq:r_positive_clean}
\end{equation}
and
\begin{equation}
r(a,b)\,\ln\!\left(\frac{r(a,b)}{r(a,b)+b}\right) > a - b\,,
\label{cor-eq:r_condition_clean}
\end{equation}
the feedback
\begin{equation}
\omega^*(a,b)=\frac{r(a,b)}{r(a,b)+b},
\label{cor-eq:omega_star_clean}
\end{equation}
with the limiting value $\omega^*(a,b)=1$ when $r=+\infty$, minimizes
\begin{equation}
J(\omega)=\int_0^\infty \Big(q(a,b)+r(a,b)\big(\omega-1-\ln\omega\big)\Big)\,dt
\label{cor-eq:J_clean}
\end{equation}
along the solutions of \eqref{eq-xi-sys}, where
\begin{equation}
q(a,b)=-a+b+r(a,b)\ln\!\left(\frac{r(a,b)}{r(a,b)+b}\right),
\label{cor-eq:q_clean}
\end{equation}
with limiting value $q(a,b)=-a$ when $r=+\infty$, has the property that 
\begin{equation}
q(a,b)>0 \quad \text{for all } (a,b)\in S\setminus\{(0,0)\}\,
\label{cor-eq:q_posdef_clean}
\end{equation}
and $q(0,0)=0$. Furthermore, both the feedback \eqref{eq:omega_star_clean} and 
\begin{equation}
\omega_0(a,b)=\frac{\ln\left(1+\dfrac{b}{r(a,b)}\right)}{\dfrac{b}{r(a,b)}},
\label{cor-eq:omega_star_clean0}
\end{equation}
with the limiting value $\omega_0(a,b)=1$ when $r=+\infty$, are stabilizing at $a=b=0$ on $S$ with respect to the CLF $V$. 
\end{corollary}

Corollary \ref{thm-5c} is one of many ways of applying Theorem \ref{thm-5}. We show another basic example, $\Theta(s)=\sqrt{s}$ studied in Section \ref{secVIII}. The design recipe yields a special case of \eqref{eq-construct-contin*}:
\begin{equation}
\label{eq-construct-contin*sqrt}
\omega^* =1+
\begin{cases}
4 \dfrac{L_fV \ L_{f+g}V}{(L_gV)^2}, 
& L_gV <0 , \\[-6pt] & L_{f+g}V>0,\\[10pt]
0, & \mbox{else.}
\end{cases}
\end{equation}

\begin{example}
\rm We return to Example \ref{ex-LinSon1} and, using \eqref{eq-construct-contin*sqrt}, with $u=\omega-1$, provide another simple alternative to \cite[(49)]{LinSontag1995}:
\begin{equation}
u^*(x)
=\begin{cases}
4x^2(1+x^2), & x>0\\
0, & x\le0
\end{cases}
\end{equation}
which guarantees stability with 
\begin{equation}
\dot V(x)
=
-|x|x^2
\begin{cases}
3+4x^2, & x>0\\
1, & x\le0
\end{cases}
\end{equation}
and optimality with $q(x)=|x|x^2$, 
\begin{equation}
r(x)=
\begin{cases}
\dfrac{1+3x^2+2x^4}{2x}, & x>0\\
+\infty, & x\le0
\end{cases}
\end{equation}
and control cost $\frac{u^2}{\left(1+\sqrt{1+u}\right)^2}$, with linear growth for $u\gg 1$.
\hfill $\triangle$
\end{example}

\section{Conclusions}

After several sections, centered around Theorem \ref{sec-inv-opt}, through which, using the predator-prey problem, we eased the reader into the structural difficulties of inverse optimal stabilization under positivity constraints on state and input, in the paper's final four sections we presented the paper's several general frameworks and constructive results. 

Theorems \ref{thm3} and \ref{thm-5} both introduce general frameworks, and produce not merely single controllers but controller families, but differ in design options. {\em If} a user succeeds in constructing some stabilizing controller and/or an associated strict CLF, then the theorems provide the following:
\begin{enumerate}
\item Theorem \ref{thm3}: Suppose a nominal controller $\omega_0$ has been designed in such a way that some expander $\Sigma$ applied to $\omega_0$ acts in the $L_gV$ direction. Then this expanded version of the nominal controller, $\omega^*=\Sigma(\omega_0)$, for every such expander $\Sigma$, is stabilizing and inverse optimal.
\item Theorem \ref{thm-5} and its Corollary \ref{thm-5cc}: Suppose the user can construct a particularly strong CLF with the property that $L_gV
\geq 0 \Rightarrow L_fV+L_gV <0$, namely, the property of {\em stability-ensuring drift-dominance when positive control is detrimental to stability and kept at the equilibrium-holding value}. Then, the user can set aside the nominal stabilizing controller used to verify this CLF condition, and employ only the  CLF itself to design, ab initio, for a contractor function $\Theta$, a whole family of inverse optimal stabilizers $\omega^*$ directly. And obtain, as a bonus, the non-optimal stabilizing ``contracted counterparts'' $\omega_0=\Theta(\omega^*)$. 
\end{enumerate}

Universal controllers in Theorems \ref{thm-univ0} and \ref{thm:universal_positive} are not within the 
general frameworks of Theorems \ref{thm3} and \ref{thm-5}. They are particularly explicit {\em single} realizations (not controller families), achieving stabilization under an unrestrictive CLF, and inverse optimality under a more restrictive CLF.

None of the general results in the paper (the single universal formulae and the general inverse optimal frameworks) are developed in the paper for {\em concrete classes of systems}, such as strict-feedback and feedforward classes in the unconstrained Euclidean setting. Tempting generalizations, for future research, are for populations  with multiple  species. 
It likely won't be easy, however, to find generalizations within 
positive systems that are open-loop unstable.



\appendix

\section{Proof of Theorem \ref{thm-intensified}}
\label{pf-Thm1}

From \eqref{Vdot-with-L-G},
and, under the baseline stabilizing feedback $U_0(X,Y)=Y^2/X$, 
\begin{eqnarray}
&&L(X,Y)+G(X,Y)U_0(X,Y)
\nonumber \\ 
&& =-q(X,Y)
= -\frac{(X-1)^2}{X}-\frac{(Y-X)^2}{X}\frac{Y}{X}
\end{eqnarray}
is negative definite on $(0,\infty)^2$ with unique zero at $(1,1)$. Writing
\begin{equation}
\dot V = L + G U = \big(L + G U_0\big) + G\,(U-U_0) = -q + G\,(U-U_0),
\end{equation}
it suffices to determine the sign of $G(U-U_0)$. Since $U=Y\,\Sigma(Y/X)$ and $U_0=Y (Y/X)$, the 
property \eqref{eq-expansion} implies $\operatorname{sign}(U-U_0)=\operatorname{sign}(\Sigma(s)-s)$, where $s=Y/X>0$ and $s\neq 1$. Based on \eqref{eq-expansion}, the definition \eqref{eq-G-def}, and $X>0$, 
$\operatorname{sign}(U-U_0)
=-\operatorname{sign}(G)$. Hence $G(U-U_0)\le 0$ for all $(X,Y)$, with equality only when 
$Y=X$. Therefore $\dot V\le -q
<0$ for all $(X,Y)\neq(1,1)$ and $\dot V(1,1)=0$, which establishes global asymptotic stability of $X=Y=1$.

\section{Proof of Lemma \ref{lem-Theta-Psi}}
\label{app-proof-Lemma1}


\emph{Pf. of A1.} For $s\neq 1$, $\Psi(s)=\exp(\cdot)>0$, and $\Psi(1)=0$ by definition.

\emph{Pf. of A2.} Differentiate the defining integrals 
\eqref{eq-Psidef1}, \eqref{eq-Psidef2}; the fundamental theorem of calculus gives $d(\ln\Psi)/ds=1/(s-\Theta(s))$.

\emph{Pf. of A3.} Let $f(s)=s-\Theta(s)$. Then $f(1)=0$,  $f'(1)=1-\Theta'(1)>0$, and  $f(s)\sim (1-\Theta'(1))(s-1)$ near $1$, so
\begin{equation}
\int_{s_0}^{s}\frac{1}{f(\tau)}\,d\tau
\sim
\frac{1}{1-\Theta'(1)}\ln|s-1|
\to -\infty
\quad\text{as } s\to 1,
\end{equation}
which implies $\Psi(s)\to 0$ and gives a strict minimum at $1$.

\emph{Pf. of A4.} From A2,
$\operatorname{sign}(\Psi'(s))
=
\operatorname{sign}\!\left(\frac{1}{s-\Theta(s)}\right)$. 
By assumption (3), $s-\Theta(s)<0$ for $s<1$ and $s-\Theta(s)>0$ for $s>1$, giving the claimed monotonicity.

\emph{Pf. of A5.} Rearrange \eqref{eq-Theta-to-Psi} as
\label{eq-def-Theta}
$\Theta(s)
=s-\frac{\Psi(s)}{\Psi'(s)}$
and differentiate 
for $s\neq 1$ to obtain
$\Theta'(s)=\frac{\Psi(s)\,\Psi''(s)}{(\Psi'(s))^2}$
Since $\Theta'(s)>0$, $\Psi(s)>0$ for $s\neq 1$, and $(\Psi'(s))^2>0$, it follows that $\Psi''(s)>0$ for all $s\neq 1$.

\vspace{-3mm}\paragraph*{Acknowledgment.}
The support by AFOSR grant FA9550-23-1-0535  and NSF grants ECCS-2151525 and CMMI-2228791 is gratefully acknowledged.

\bibliography{sample.bib}

@article{DeLeenheerAeyels2001,
  author  = {Patrick De Leenheer and Dirk Aeyels},
  title   = {Stabilization of positive linear systems},
  journal = {Systems \& Control Letters},
  volume  = {44},
  number  = {4},
  pages   = {259--271},
  year    = {2001},
  issn    = {0167-6911},
  doi     = {10.1016/S0167-6911(01)00146-3}
}

@article{AngeliSontag2003,
  author  = {Angeli, David and Sontag, Eduardo D.},
  title   = {Monotone Control Systems},
  journal = {IEEE Transactions on Automatic Control},
  year    = {2003},
  volume  = {48},
  number  = {10},
  pages   = {1684--1698},
  month   = oct,
  doi     = {10.1109/TAC.2003.818723}
}

@article{BlanchiniColaneriValcher2020,
  author  = {Blanchini, Franco and Colaneri, Paolo and Valcher, Marco E.},
  title   = {Co-positive Lyapunov functions for the stabilization of positive switched systems},
  journal = {Automatica},
  year    = {2020},
  volume  = {113},
  pages   = {108770},
  doi     = {10.1016/j.automatica.2020.108770}
}

@article{Rantzer2015,
  author  = {Rantzer, Anders},
  title   = {Scalable control of positive systems},
  journal = {European Journal of Control},
  year    = {2015},
  volume  = {24},
  pages   = {72--80},
  doi     = {10.1016/j.ejcon.2014.12.005}
}

@inproceedings{PapageorgiouHajSlSalemBlosseville1991,
  author    = {Papageorgiou, Markos and Haj-Salem, Habib and Blosseville, Jean-Marc},
  title     = {ALINEA: A local feedback control law for on-ramp metering},
  booktitle = {Transportation Research Record},
  year      = {1991},
  number    = {1320},
  pages     = {58--64}
}

@article{PapageorgiouReview2018,
  author  = {Papageorgiou, Markos and Papamichail, Ioannis and Messmer, Alexander and Wang, Yibing},
  title   = {Traffic-responsive linked ramp-metering control},
  journal = {Annual Reviews in Control},
  year    = {2018},
  volume  = {45},
  pages   = {278--305},
  doi     = {10.1016/j.arcontrol.2018.08.001}
}

@article{RamirezVanDerSchaftMaschke2013,
  author  = {Ramirez, Hector and van der Schaft, Arjan J. and Maschke, Bernhard M.},
  title   = {Irreversible port-Hamiltonian systems: A generalization of energy-based modeling to nonequilibrium thermodynamics},
  journal = {Chemical Engineering Science},
  year    = {2013},
  volume  = {92},
  pages   = {171--185},
  doi     = {10.1016/j.ces.2013.01.038}
}

@article{sontag1989universal,
  title={A ‘universal’construction of Artstein's theorem on nonlinear stabilization},
  author={Sontag, Eduardo D},
  journal={Systems \& control letters},
  volume={13},
  number={2},
  pages={117--123},
  year={1989},
  publisher={Elsevier}
}

@ARTICLE{661589,
  author={Krstic, M. and Zhong-Hua Li},
  journal={IEEE Transactions on Automatic Control}, 
  title={Inverse optimal design of input-to-state stabilizing nonlinear controllers}, 
  year={1998},
  volume={43},
  number={3},
  pages={336-350},
  doi={10.1109/9.661589}}

@article{10.1115/1.3653115,
    author = {Kalman, R. E.},
    title = {When Is a Linear Control System Optimal?},
    journal = {Journal of Basic Engineering},
    volume = {86},
    number = {1},
    pages = {51-60},
    year = {1964},
    month = {03},
    issn = {0021-9223},
    doi = {10.1115/1.3653115},
    url = {https://doi.org/10.1115/1.3653115},
    eprint = {https://asmedigitalcollection.asme.org/fluidsengineering/article-pdf/86/1/51/5763755/51_1.pdf},
}

@article{KrsticAsymmetryDemystified,
  author       = {Miroslav Krsti{\'c}},
  title        = {Asymmetry Demystified: Strict CLFs and Feedbacks for Predator–Prey Interconnections},
  journal      = {IEEE Transactions on Automatic Control},
  note         = {Under review},
  year         = {2026},
  url          = {https://arxiv.org/abs/2602.21594},
}

@book{Krstic1998Stabilization,
  author    = {Miroslav Krsti\'{c} and Hua Deng},
  title     = {Stabilization of Nonlinear Uncertain Systems},
  publisher = {Springer},
  series    = {Communications and Control Engineering},
  year      = {1998},
  isbn      = {9781852330200},
  address   = {London, UK}
}

@article{LinSontag1995,
  author       = {Yuandan Lin and Eduardo D. Sontag},
  title        = {Control-Lyapunov Universal Formulas for Restricted Inputs},
  journal      = {Control: Theory and Advanced Technology},
  volume       = {10},
  number       = {4, part 5},
  pages        = {1981--2004},
  year         = {1995},
  publisher    = {World Scientific},
  doi          = {},
  note         = {This article generalizes universal CLF formulas to systems with bounded and/or positive controls.}
}

@INPROCEEDINGS{Steentjes2018Feedback,
  author={Steentjes, Tom R. V. and Doban, Alina and Lazar, Mircea},
  booktitle={2018 17th European Control Conference (ECC)},
  title={Feedback stabilization of positive nonlinear systems with applications to biological systems},
  year={2018},
  pages={2684-2689},
  doi={10.23919/ECC.2018.8550371},
  organization={IEEE}
}

@ARTICLE{9740520,
  author={Ito, Hiroshi},
  journal={IEEE Transactions on Automatic Control}, 
  title={Strict Smooth Lyapunov Functions and Vaccination Control of the SIR Model Certified by ISS}, 
  year={2022},
  volume={67},
  number={9},
  pages={4514-4528},
  doi={10.1109/TAC.2022.3161395}}

@ARTICLE{11080060,
  author={Veil, Carina and Krstić, Miroslav and Karafyllis, Iasson and Diagne, Mamadou and Sawodny, Oliver},
  journal={IEEE Transactions on Automatic Control}, 
  title={Stabilization of Predator–Prey Age-Structured Hyperbolic PDE When Harvesting Both Species is Inevitable}, 
  year={2026},
  volume={71},
  number={1},
  pages={123-138},
  doi={10.1109/TAC.2025.3589108}}

@article{malisoff2016stabilization,
  title={Stabilization and robustness analysis for a chain of exponential integrators using strict Lyapunov functions},
  author={Malisoff, Michael and Krstic, Miroslav},
  journal={Automatica},
  volume={68},
  pages={184--193},
  year={2016},
  publisher={Elsevier}
}
\bibliographystyle{IEEEtranS}

\end{document}